\documentclass[reqno,11pt]{amsart}
\usepackage{amsmath,amsfonts,mathrsfs,amssymb}
\allowdisplaybreaks[4]
\newtheorem{theo}{Theorem}[section]
\newtheorem{prop}[theo]{Proposition}
\newtheorem{lem}[theo]{Lemma}
\newtheorem{cor}[theo]{Corollary}

\newtheorem{defi}{Definition}
\newtheorem{rem}{Remark}
\def\tr{\mathop{\rm tr}\nolimits}

\newcommand{\N}{{\mathbb N}}
\newcommand{\Z}{{\mathbb Z}}
\newcommand{\R}{{\mathbb R}}

\begin{document}
\title[Fractal dimensions of Spectrum]
{The fractal dimensions of the spectrum of
Sturm Hamiltonian}
\author{
Qing-Hui LIU, Yan-hui Qu, Zhi-Ying WEN }


\maketitle

\begin{abstract}
Let $\alpha\in(0,1)$ be irrational and $[0;a_1,a_2,\cdots]$ be the continued fraction expansion of $\alpha$.
Let $H_{\alpha,V}$ be the Sturm Hamiltonian with
frequency $\alpha$ and coupling $V$, $\Sigma_{\alpha,V}$ be the spectrum of $H_{\alpha,V}$.
The fractal dimensions of the spectrum have been determined by Fan, Liu and Wen (Erg. Th. Dyn. Sys.,2011) when $\{a_n\}_{n\ge1}$  is bounded.
The present paper will treat the most difficult case, i.e, $\{a_n\}_{n\ge1}$  is unbounded. We prove that
for $V\ge24$,
$$
\dim_H\ \Sigma_{\alpha,V}=s_*(V)\ \ \ \text{ and }\ \ \  \overline{\dim}_B\ \Sigma_{\alpha,V}=s^*(V),
$$
where $s_*(V)$ and $s^*(V)$ are lower and upper pre-dimensions respectively. By this result, we determine the fractal dimensions of the spectrums
for all Sturm Hamiltonians.

We also show      the following results: $s_*(V)$ and $s^*(V)$ are Lipschitz continuous
on any bounded interval of $[24,\infty)$; the limits
$s_*(V)\ln V$ and $s^*(V)\ln V$ exist as $V$ tend to infinity,
and the limits are constants  only depending on $\alpha$; $s^\ast(V)=1$ if and only if $\limsup_{n\to\infty}(a_1\cdots a_n)^{1/n}=\infty,$ which can be compared with the fact:  $s_\ast(V)=1$ if and only if $\liminf_{n\to\infty}(a_1\cdots a_n)^{1/n}=\infty$( Liu and Wen,  Potential anal. 2004).


\end{abstract}

\section{Introduction}

The Sturm Hamiltonian is a  discrete Schr\"odinger operator
$$
(H\psi)_n:=\psi_{n-1}+\psi_{n+1}+v_n \psi_n
$$
on $\ell^2(\Z),$ where
the potential $(v_n)_{n\in\mathbb{Z}}$ is given by
\begin{equation}\label{sturm}
v_n=V\chi_{[1-\alpha,1)}(n\alpha+\phi \mod 1),\quad \forall n\in\mathbb{Z},
\end{equation}
where $\alpha\in(0,1)$ is  irrational, and is called frequency,
$V>0$ is called  coupling, $\phi\in[0,1)$ is
called phase.  It is known that the spectrum of Sturm Hamiltonian
 is independent of $\phi$, so we take $\phi=0$ and denote the spectrum by $\Sigma_{\alpha,V}$.
We often simplify the notation $\Sigma_{\alpha,V}$ to $\Sigma_V$ or $\Sigma$ when $\alpha$ or $V$ are fixed.
The present paper is devoted to  determine  the fractal dimensions of $\Sigma_{\alpha,V}$ for all irrational $\alpha$.

The most prominent model among the Sturm Hamiltonian is the Fibonacci Hamiltonian, which is given by taking $\alpha$ to be the golden number $\alpha_0:=(\sqrt{5}-1)/2.$ This model was  introduced by physicists  to model the quasicrystal  system(\cite{KKT,OPRSS}).
S\"ut\"o  showed that the spectrum always has zero Lebesgue measure \cite{Su},
$$
L(\Sigma_{\alpha_0,V})=0,\ \ \ \ \text{ for all } V>0.
$$

Then it is  natural  to ask what is the fractal dimension of the spectrum.
Raymond first estimated the Hausdorff dimension \cite{R}, and he showed that $\dim_H\Sigma_{\alpha_0,V}<1$ for $V>4$.  Jitomirskaya and Last \cite{JL} showed that for any $V>0$, the spectral measure of the operator has positive Hausdorff dimension, as a consequence $\dim_H\Sigma_{\alpha_0,V}>0.$ By using dynamical method, Damanik et al. \cite{DEGT} showed that if $V\ge 16$ then
\begin{equation}\label{coincide-H-B}
\dim_B\Sigma_{\alpha_0,V}=\dim_H \Sigma_{\alpha_0,V}.
\end{equation}
They also got lower and upper bounds for the dimensions. Due to these bounds they further showed that
\begin{equation}\label{asym-Fibo}
\lim_{V\to\infty} \dim_H \Sigma_{\alpha_0,V}\ln  V =\ln (1+\sqrt{2}).
\end{equation}
We remark that more than a natural question, the fractal dimensions of the spectrum are also related to the rates of propagation of the fastest part of the wavepacket(see  \cite{DEGT} for detail).

Write $d(V)=\dim_H \Sigma_{\alpha_0,V}$. Cantat \cite{C}, Damanik and Gorodetski \cite{DG} showed that: $d(V)\in(0,1)$ is analytic on $(0,\infty)$.
In \cite{DG2}, Damanik and Gorodetski further showed that $\lim_{V\downarrow 0}d(V)=1$ and the speed is linear.

\bigskip

Now we  go back to the general Sturm Hamiltonian case. We fix an irrational $\alpha\in(0,1)$  with continued fraction expansion $[0;a_1,a_2,\cdots]$. Write
\begin{equation}\label{K-u-l}
K_\ast(\alpha)=\liminf_{k\rightarrow\infty}
(\prod_{i=1}^k a_i)^{1/k}\ \text{ and }\  K^\ast(\alpha)=
\limsup_{k\rightarrow\infty}(\prod_{i=1}^k a_i)^{1/k}.
\end{equation}

 Bellissard et al.  \cite{BIST} showed
 that $\Sigma_{\alpha,V}$ is a Cantor set of
Lebesgue measure zero.
Damanik, Killip and Lenz \cite{DKL} showed that, if   $\limsup\limits_{k\rightarrow\infty}\frac{1}{k}\sum_{i=1}^k
a_i<\infty$, then  $\dim_H \Sigma_{\alpha,V}>0$, notice that the set of such $\alpha$ has Lebesgue measure 0 in $(0,1).$ Basing on the analysis of Raymond \cite{R} about  the structure of spectrum, Liu and Wen \cite{LW} showed that  for $V\ge 20$
\begin{equation}\label{dim-hausdorff}
\begin{cases}
\dim_H \Sigma_{\alpha,V} \in(0,1) & \text{ if } K_\ast(\alpha)< \infty\\
\dim_H \Sigma_{\alpha,V} =1 & \text{ if } K_\ast(\alpha)= \infty.
\end{cases}
\end{equation}

Raymond \cite{R}, Liu and Wen \cite{LW} showed that the spectrum $\Sigma_{\alpha,V}$ has a natural covering structure. This structure makes it possible to define the so called pre-dimensions $s_\ast(V)$ and $s^\ast(V)$(see \eqref{predim} for the definition).
Liu, Peyriere and Wen \cite{LPW07} showed that
\begin{equation}\label{dimineq}
{\dim}_H \Sigma_{\alpha,V}\le s_*(V),\quad \overline{\dim}_B
\Sigma_{\alpha,V} \ge s^*(V).
\end{equation}
 Moreover, they show that, for $\alpha$ of {\it bounded type}, i.e. $\{a_k\}_{k\ge1}$ bounded
\begin{equation*}
\lim_{V\to \infty} s_*(V)\ln V = -\ln f_*(\alpha),\quad \lim_{V\to
\infty} s^*(V)\ln V = -\ln f^*(\alpha).
\end{equation*}
(see \eqref{f-u-l-alpha} for the definition of $f_*(\alpha)$ and $f^*(\alpha)$).
When $\alpha=\alpha_0$ they proved that
$$
f_\ast(\alpha_0)=f^\ast(\alpha_0)=(1+\sqrt{2})^{-1}.
$$

Recently Fan, Liu and Wen \cite{FLW} showed that  for $\alpha$ of bounded type, the two inequalities in \eqref{dimineq} are indeed equalities.  Moreover if  $\{a_k\}_{k\ge 1}$ is eventially periodic, then $s_\ast(V)=s^\ast(V).$ Thus for $\alpha$ of bounded type, they determined the fractal dimensions of the spectrum and generalized \eqref{coincide-H-B} and
\eqref{asym-Fibo}.

In this paper we will complete the picture for the fractal dimensions of the spectrum of Sturm Hamiltonian by treating the most difficult part: $\alpha$ is of unbounded type,
i.e., $\{a_k\}_{k\ge 1}$ is unbounded. We state now the main results of the paper and some remarks.

\begin{theo}\label{fracdim}
Let $V\ge 24$,  and $\alpha\in(0,1)$ be irrational. Then
\begin{equation}\label{formula-dim}
\dim_H \Sigma_{\alpha,V}=s_*(V)
\ \ \ \text{ and }\ \ \  \overline{\dim}_B \Sigma_{\alpha,V}=s^*(V).
\end{equation}
Moreover
\begin{equation}\label{dim-box}
\begin{cases}
\overline{\dim}_B \Sigma_{\alpha,V}\in(0,1) & \text{ if } K^\ast(\alpha)<\infty\\
\overline{\dim}_B \Sigma_{\alpha,V} =1 & \text{ if } K^\ast(\alpha)= \infty.
\end{cases}
\end{equation}
\end{theo}

\begin{theo}\label{limfor}
Fix $\alpha\in(0,1)$ irrational.  Let $f_*(\alpha)$ and $f^*(\alpha)$
be defined as in \eqref{f-u-l-alpha}, then
\begin{equation}\label{asym-general}
\lim_{V\to \infty} s_*(V)\ln V
= -\ln f_*(\alpha),\quad
\lim_{V\to \infty} s^*(V)\ln V
= -\log f^*(\alpha).
\end{equation}
\end{theo}

\begin{theo}\label{lip-conti}
$s_*(V)$ and $s^*(V)$ are
  Lipschitz continuous on any bounded interval of $[24,\infty).$
\end{theo}

\begin{rem}
{\rm
 1) Formula \eqref{dim-box}  is  the box dimension counterpart of \eqref{dim-hausdorff}, and the formulas \eqref{dim-hausdorff} and \eqref{dim-box} give the sufficient and necessary conditions such that Hausdorff dimensions and box dimension are strictly less than 1 and positive.

 2) In general we can not expect $s_\ast(V)=s^\ast(V).$ The simplest example is  as follows: take $\alpha=[0;a_1,a_2,\cdots]$ such that
$$
K_\ast(\alpha)=1 \ \ \ \text{ and }\ \ \  K^\ast(\alpha) =\infty.
$$
Then by \eqref{dim-hausdorff} and \eqref{formula-dim} we have $s_\ast(V)<1$,  by \eqref{dim-box} and \eqref{formula-dim} we have $s^\ast(V)=1.$

3) Formula \eqref{asym-general} is a complete generalization of \eqref{asym-Fibo}.

4) We know that in  the Fibonacci case, the dimension function $d(V)$ is real analytic (\cite{C,DG}). For the Sturm case, we can not expect such strong regularity. However by Theorem \ref{lip-conti}, both Hausdorff and Box dimension functions are still Lipschitz continuous, which will be obtained essentially from the formula \eqref{formula-dim}.
}
\end{rem}

We will compare the present work with some previous works \cite{LW, LPW07, FLW} to explain  the main difficulties we will meet
and indicate some new ideas and techniques we will introduce.

The main idea in \cite{LW} is essentially introducing a natural covering structure by construct spectral generating bands,
and estimate the length of spectral generating bands by computing one-order derivative of spectral generating polynomial.
The key points in \cite{FLW} consists of, on the one hand,  generalizing the  Cookie-Cutter-like structure  introduced by Ma, Rao and Wen \cite{MRW} and developing some related
techniques for establishing the Gibbs like measure; and on the other hand, giving a more exact formula for the derivative of spectral generating polynomial,
and estimating of the two-order derivative.

But if $\{a_k\}_{k\ge 1}$ is  unbounded, these techniques and methods are not enough. To see this,  we recall first  the definition of Cookie-Cutter set. Taking $I=[0,1]$, $I_0,I_1\subset I$ be two disjoint
subintervals of $I$, let $f:I_0\cup I_1\rightarrow I$ satisfies

(C-i)\ \ $f|_{I_0}$, $f|_{I_1}$ are $1-1$ mappings onto $I$;

(C-ii)\ \ $C^{1+\gamma}$ H\"older($\gamma>0$), i.e., $\exists c>0$,
$$
|f'(x)-f'(y)|\le c |x-y|^{\gamma},\quad \forall x,y\in I_0\cup I_1;
$$

(C-iii)\ \ expansion, i.e. there exist $B>b>1$, for any $x\in I_0\cup
I_1$,
$$1<b\le|f'(x)|\le B<\infty.$$
We call $f$ a Cookie-Cutter map. The hyperbolic attractor of $f$ is defined as
\begin{equation}\label{CC}
E:=\{x\in\mathbb{R}\ |\ \forall k\ge0, \
f^k(x)\in[0,1]\}.
\end{equation}
 $E$ is called the Cookie-Cutter set associated with the Cookie-Cutter map $f$.

Let
$\phi_0=(f|_{I_0})^{-1}$,$\phi_1=(f|_{I_1})^{-1}$ and
$\Sigma=\{0,1\}$. For any $k\ge1$, $\sigma=i_1 i_2\cdots
i_k\in\Sigma^k$, define
$I_\sigma=\phi_{i_1}\circ\phi_{i_2}\circ\cdots\circ\phi_{i_k}(I)$,
then $f^k(I_\sigma)=I$ and
$E=\bigcap_{k\ge1}\bigcup_{\sigma\in\Sigma^k}I_\sigma$.

As in \cite{Fa}, the system
satisfies the principle of bounded variation, i.e., there exists
$\xi\ge1$ such that, for any $k\ge1$,  $\sigma\in\Sigma^k$, and
any $x,y\in I_\sigma$,
$$|(f^k)'(x)/(f^k)'(y)|\le\xi;$$  and the system also satisfies  the principle of bounded distortion,
i.e. for any $x\in I_\sigma$,
$$\xi^{-1}\le|(f^k)'(x)|\,|I_\sigma|\le\xi.$$
Notice that  by the chain rule, we have
\begin{equation}\label{comp}
(f^k)'(x)=f'(f^{k-1}(x))f'(f^{k-2}(x))\cdots f'(x).
\end{equation}
By these two principles, we see that the length of the interval $I_\sigma$ could be
estimated by the derivative of $f^k$ at any point of $I_\sigma$.

Moreover, Ma, Rao and Wen \cite{MRW} showed that the system also satisfies
 the principle of bounded covariation, i.e., for any $m>k>0$,
 $\sigma_1,\sigma_2\in\Sigma^k$, and $\tau\in\Sigma^{m-k}$,
$$\frac{|I_{\sigma_1*\tau|}}{|I_{\sigma_1}|}\le\xi^2
\frac{|I_{\sigma_2*\tau|}}{|I_{\sigma_2}|}.$$

With these principles, one can prove  the existence of the Gibbs measure, i.e., for any $0<\beta<1$,
there exists probability measure $\mu_\beta$ such that, for any
$k>0$ and $\sigma\in\Sigma^k$,
$$\xi^{-2}\frac{|I_\sigma|^\beta}{\sum_{\tau\in\Sigma^k}|I_\tau|^\beta}\le
\mu_\beta(I_\sigma)
\le\xi^2\frac{|I_\sigma|^\beta}{\sum_{\tau\in\Sigma^k}|I_\tau|^\beta}.$$
These measures are crucial for analyzing  fractal dimensions of the attractors, such as  formulas for Hausdorff dimension, box dimension and
continuous dependence  of dimensions with respect to  $f$.

Now we turn to the Cookie-Cutter-like set introduced by Ma, Rao and Wen (\cite{MRW}) which generalizes the classical
Cookie-Cutter set:
\begin{equation}\label{CClike}
E=\{x\in\mathbb{R}\ |\ \forall k\ge0, f_k\circ f_{k-1}\circ\cdots\circ f_1(x)\in[0,1]\},
\end{equation}
where for any $k\ge1$, $f_k$ satisfies

(U-i)\ \ $\exists I_j^k\subset I=[0,1]$, $j=1,2,\cdots,m_k$, mutually disjoint, such that
$f_k|_{I_j^k}$ are $1-1$ mappings onto $I$;

(U-ii)\ \ $C^{1+\gamma}$ H\"older($\gamma>0$), i.e., $\exists c_k>0$,
$$
|f_k'(x)-f_k'(y)|\le c_k |x-y|^{\gamma},\quad \forall x,y\in \bigcup_j I_j^k;
$$

(U-iii)\ \ expansion, i.e. there exist $B_k>b_k>1$, for any $x\in \bigcup_j I_j^k$,
$$1<b_k\le|f_k'(x)|\le B_k<\infty.$$

Comparing with \eqref{CC}, we see that the $k$-th iteration of
the same mapping $f$ is replaced by composition of $k$ different
mappings in \eqref{CClike}.

Under the conditions of uniformly H\"older
and uniformly bounded expansion, i.e.,
\begin{equation}\label{bounded}
\sup c_k<\infty,\quad 1<\inf b_k\le \sup B_k<\infty,
\end{equation}
 the principles of bounded variation, bounded distortion, bounded covariation
and the existence of Gibbs like measure were proven in \cite{MRW}. They gave  formulas for the dimensions and showed the  continuous dependence  of
dimensions with respect to  $\{f_k\}_{k\ge1}$.

In \cite{FLW}, to study the dimensional property of spectrum with bounded type,
they apply the technique of Cookie-Cutter-like set in the following way.
For every spectral generating band $B$, there is a generating polynomial
$h_{B}$ such that $h_{B}$ is monotone on $B$ and $h_B(B)=[-2,2]$.
They estimated the length of $B$ by help of $h_B$. Suppose $(B_k)_{k=0}^n$
is a sequence of spectral generating bands of order from $0$ to $n$ with
$$B_n\subset B_{n-1}\subset \cdots B_0,$$
and suppose their corresponding generating polynomials are
$(h_i)_{i=0}^n$. Noting that $h_0'=1$, and
$$h_n'=\frac{h_n'}{h_{n-1}'}\frac{h_{n-1}'}{h_{n-2}'}\cdots\frac{h_{1}'}{h_{0}'}.$$
Comparing with \eqref{comp}, they analyze ${h'_{k+1}}/{h'_{k}}$
in stead  of analyzing $f'(f^k(x))$.
Analogous to condition (U-iii), they proved
\begin{equation}\label{uii}
4<{h_{k}'(x)}/{h_{k-1}'(x)}<B_k
\end{equation}
And instead of H\"older condition (U-ii), they proved (see also \eqref{prop2})
\begin{equation}\label{holder}
\begin{array}{rcl}
\left|\frac{h'_{k+1}(x)}{h'_{k}(x)}-
\frac{h'_{k+1}(y)}{h'_{k}(y)}\right|
&\le& t_k(|h_{k}(x)-h_{k}(y)|+\frac{d_k}{6}|h_{k-1}(x)-h_{k-1}(y)|)\\
&&+\frac{1}{6e_k}
\left|\frac{h'_{k}(x)}{h'_{k-1}(x)}-\frac{h'_{k}(y)}{h'_{k-1}(y)}\right|.
\end{array}\end{equation}

Notice that all parameters $B_k, t_k, d_k, e_k$ in \eqref{uii} and \eqref{holder}
depend on $a_k$. If $\{a_k\}$ is bounded, $B_k, t_k, d_k, e_k$ are also bounded,
thus they can apply techniques of \cite{MRW} directly.

But if the sequence $\{a_n\}$ is unbounded, then
$\sup_k t_k=\infty,\ \sup_k B_k=\infty$.
Return to the Cookie-Cutter case, comparing with \eqref{bounded},
this is equivalent to
$$\sup_k c_k=\infty,\quad \sup_k B_k=\infty,$$
i.e., neither  uniformly H\"older nor uniformly
bounded expansion.
By carefully  analyzing the relation between $c_k$ and $B_k$,
we find that the conclusion of \cite{MRW} still holds
if we relax the condition \eqref{bounded} to
require that for some constant $C>0$,
$$c_k\le C\cdot\inf_{x\in\bigcup_j I_j^k} |f_k'(x)|,\quad\forall k>0.$$
By this way, we could overcome the difficulty $c_k$ and $B_k$ not  bounded.

This technique can be accommodated to our case to  show the principles of bounded variation, distortion and covariation for the spectrum,
by making more accurate estimations for
$t_k,d_k,e_k$ and $h_k(x)-h_k(y)$ in \eqref{holder}.

It is more tricky to construct a  Gibbs like measure, since
when  $\{a_k\}_{k\ge 1}$ is unbounded,
we can not distribute mass evenly on different types of spectral generating bands.
However, with  much effort,  we can still construct a weak type Gibbs like measure
which plays the  same role  as Gibbs like measure.

Finally, in applying mass distribution principle to get a good lower bound for Hausdorff dimension, we will meet the following difficulty: for a spectral generating band $B$
of order $k$ and  type $III$,
it contains $a_k$ spectral generating bands (denote as $B_l$ for $1\le l\le a_k$)
of order $k+1$ and  type $I$ with contraction ratios
$$|B_l|/|B|\sim a_k^{-1}\sin^2 \frac{l\pi}{a_k+1}, \quad l=1,2,\cdots, a_k.$$
The contraction ratios vary from  $a_k^{-1}$ to $a_k^{-3}$,
so the weak Gibbs like measure fails to give desired
estimation.

To overcome this difficulty, we introduce
a truncation technique.
For any small $\varepsilon>0$, we delete the intervals $B_l$ with
$$0<l/(a_k+1)<\varepsilon\quad \mbox{or}\quad 1-\varepsilon<l/(a_k+1)<1.$$
So the remaining intervals satisfy
$$|B_l|/|B|\gtrsim\varepsilon^2 a_k^{-1}.$$
Denote the remaining set by $E_\varepsilon$.
Now we can apply weak Gibbs like measure supported on $E_\varepsilon$
to get lower bound of Hausdorff dimension for $E_\varepsilon$
(here we use idea from \cite{FWW}),
and then obtain the desired lower bound for $E$ as $\varepsilon$ tends to $0$.

\medskip

The plan of the paper is  as follows.
In Section \ref{structure}, we will study the structure of the spectrum,
especially  we will give a coding for the spectrum. In Section \ref{prepare},
we state some results which  we need to prove the main Theorems.
Section \ref{thm1} and Section \ref{thm2} are devoted to the proofs of Theorem \ref{fracdim}
and Theorem \ref{limfor} respectively. The proof of Theorem \ref{lip-conti} will be postponed to Section \ref{bdcov}
since the proof of which need a technique lemma.
The rest sections are devoted to the proofs of the results stated in section \ref{prepare}.

\section{The  structure and coding of the spectrum }\label{structure}

We describe  the structure of the spectrum
$\Sigma=\Sigma_{\alpha,V}$ for some fixed $\alpha$ and $V$. We will see that $\Sigma$  has a natural covering structure which can be associated with a natural coding.

Let $\alpha=[0;a_1,a_2,\cdots]\in(0,1)$ be  irrational,
let $p_k/q_k$$(k>0)$ be the $k$-th partial quotient  of $\alpha$
given by:
$$\begin{array}{l}
p_{-1}=1,\quad p_0=0,\quad p_{k+1}=a_{k+1} p_k+p_{k-1},\ k\ge 0,\\
q_{-1}=0,\quad q_0=1,\quad q_{k+1}=a_{k+1} q_k+q_{k-1},\ k\ge 0.
\end{array}$$

\smallskip

\noindent Let $k\geq1$ and $x\in\mathbb{R}$, the transfer matrix $M_k(x)$
over $q_k$ sites is defined by
$${\mathbf M}_k(x):=
\left[\begin{array}{cc}x-v_{q_k}&-1\\ 1&0\end{array}\right]
\left[\begin{array}{cc}x-v_{q_k-1}&-1\\ 1&0\end{array}\right]
\cdots
\left[\begin{array}{cc}x-v_1&-1\\ 1&0\end{array}\right],$$
where $v_n$ is defined in \eqref{sturm}. By convention we  take
$$\begin{array}{l}
{\mathbf M}_{-1}(x)= \left[\begin{array}{cc}1&-V\\
0&1\end{array}\right],\quad {\mathbf M}_{0}(x)=
\left[\begin{array}{cc}x&-1\\ 1&0\end{array}\right].
\end{array}$$

\smallskip

For $k\ge0$, $p\ge-1$, let $t_{(k,p)}(x)=\tr {\mathbf M}_{k-1}(x) {\mathbf M}_k^p(x)$ and
$$
\sigma_{(k,p)}=\{x\in\mathbb{R}:|t_{(k,p)}(x)|\leq2\}
$$ where  $\tr M$
stands for the trace of the matrix $M$.

With these notations, we
collect some known facts that will be used later, for more
details, we refer to \cite{BIST,R,Su,T}.
\begin{itemize}
\item[(A)]\ Renormalization relation.
For any $k\ge0$
$${\mathbf M}_{k+1}(x)={\mathbf M}_{k-1}(x)({\mathbf
M}_k(x))^{a_{k+1}},
$$
so, $t_{(k+1,0)}=t_{(k-1,a_{k})}$, $t_{(k,-1)}=t_{(k-1,a_k-1)}$.

\item[(B)]\ Structure of $\sigma_{(k,p)}(k\ge0,p\ge-1)$.
For $V>0$, $\sigma_{(k,p)}$ is made of $\deg t_{(k,p)}$ disjoint
closed intervals.

\item[(C)]\ Trace relation.
Defining $\Lambda(x,y,z)=x^2+y^2+z^2-xyz-4$, then
$$\Lambda(t_{(k+1,0)},t_{(k,p)},t_{(k,p+1)})=V^2.
$$
Thus for any $k\in\mathbb{N}$, $p\geq 0$ and $V>4$,
\begin{equation}\label{empty}
\sigma_{(k+1,0)}\cap \sigma_{(k,p)}\cap\sigma_{(k,p-1)}=\emptyset.
\end{equation}

\item[(D)]\ Covering property.
For any $k\ge0$, $p\ge-1$,
\begin{equation}\label{tao}
\sigma_{(k,p+1)}\subset \sigma_{(k+1,0)}\cup \sigma_{(k,p)},
\end{equation}
then
$$(\sigma_{(k+2,0)}\cup\sigma_{(k+1,0)})\subset
(\sigma_{(k+1,0)}\cup\sigma_{(k,0)}).$$
Moreover
$$\Sigma=\bigcap_{k\ge0}(\sigma_{(k+1,0)}\cup\sigma_{(k,0)}).$$
\end{itemize}

The intervals of $\sigma_{(k,p)}$ will be called the {\em bands},
when we discuss only one of these bands, it is often denoted as
$B_{(k,p)}$. Property (B) also implies $t_{(k,p)}(x)$ is monotone on
$B_{(k,p)}$, and
$$t_{(k,p)}(B_{(k,p)})=[-2,2],$$
we call $t_{(k,p)}$ the {\em generating polynomial} of $B_{(k,p)}$.

$\{\sigma_{(k+1,0)}\cup\sigma_{(k,0)}:k\ge 0\}$ form a covering of $\Sigma$.
However there are some repetitions between $\sigma_{(k,0)}\cup\sigma_{(k-1,0)}$
and $\sigma_{(k+1,0)}\cup\sigma_{(k,0)}$.
It is possible to choose a coverings of $\Sigma$ elaborately such that
we can get rid of these repetitions, as we will describe in the follows:

\begin{defi}{\rm (\cite{R,LW})}\label{def1}
For $V>4$, $k\ge0$, we define three types of bands as follows:

$(k,{\rm I})$-type band: a band of $\sigma_{(k,1)}$ contained in a
band of $\sigma_{(k,0)}$;

$(k,{\rm II})$-type band: a band of $\sigma_{(k+1,0)}$ contained
in a band of $\sigma_{(k,-1)}$;

$(k,{\rm III})$-type band: a band of $\sigma_{(k+1,0)}$ contained
in a band of $\sigma_{(k,0)}$.
\end{defi}

By the property (B),  \eqref{empty} and \eqref{tao}, all three kinds
of types of bands are well defined, and we call these bands {\em
spectral generating bands of order $k$} (the type I band is called
the type I gap in \cite{R}). Note that for order $0$, there is only
one $(0,{\rm I})$-type band $\sigma_{(0,1)}=[V-2,V+2]$ (the
corresponding generating polynomial is $t_{(0,1)}=x-V$), and only
one $(0,{\rm III})$ type band $\sigma_{(1,0)}=[-2,2]$ (the
corresponding generating polynomial is $t_{(1,0)}=x$). They are
contained in $\sigma_{(0,0)}=(-\infty,+\infty)$ with corresponding
generating polynomial $t_{(0,0)}\equiv2$. For the convenience, we
call $\sigma_{(0,0)}$ the spectral generating band of order $-1$.

\smallskip

For any $k\ge-1$, denote by $\mathscr{G}_k$ the set of all spectral
generating bands of order $k$, then the intervals in $\mathscr{G}_k$ are disjoint.
Moreover (\cite{LW,FLW})
\begin{itemize}
\item $(\sigma_{(k+2,0)}\cup\sigma_{(k+1,0)})\subset
\bigcup_{B\in\mathscr{G}_k}B
\subset (\sigma_{(k+1,0)}\cup\sigma_{(k,0)})$,
thus
\begin{equation}\label{struc-spec}
\Sigma=\bigcap_{k\ge0}
\bigcup_{B\in\mathscr{G}_k}B;
\end{equation}
\item
any $(k,I)$-type band contains only one band in $\mathscr{G}_{k+1}$, which is of  $(k+1,II)$-type.
\item
any $(k,II)$-type band contains $2a_{k+1}+1$ bands in $\mathscr{G}_{k+1}$,
$a_{k+1}+1$ of which are of  $(k+1,I)$-type and $a_{k+1}$ of which are of  $(k+1,III)$-type.
\item
any $(k,III)$-type band contains $2a_{k+1}-1$ bands in $\mathscr{G}_{k+1}$,
$a_{k+1}$ of which are of  $(k+1,I)$-type and $a_{k+1}-1$ of which are of  $(k+1,III)$-type.
\end{itemize}

Thus $\{\mathscr{G}_k\}_{k\ge0}$ forms a natural
covering(\cite{LW05,LPW07}) of the spectrum $\Sigma$. For any $k\ge1$, let $s_k$ be the unique real number in $[0,1]$ satisfies
$$\sum_{B\in\mathscr{G}_k}|B|^{s_k}=1,$$
and define the pre-dimensions of $\Sigma $ by
\begin{equation}\label{predim}
s_*(V)=\liminf_{k\rightarrow\infty}s_k,\quad
s^*(V)=\limsup_{k\rightarrow\infty}s_k.
\end{equation}

In the following we will give a coding for $\Sigma$.
Let
$$
\mathcal E:=\{(I,II),(II,I),(II,III),(III,I),(III,III)\}
$$
be the admissible edges.
To simplify the notation, we write
$$
e_{12}=(I,II), e_{21}=(II,I),e_{23}=(II,III), e_{31}=(III,I), e_{33}=(III,III).
$$
For each $n\in\N$
define
$$
\tau_{e}(n)=
\begin{cases}
1& e=e_{12}\\
n+1& e=e_{21}\\
n& e=e_{23}\\
n& e=e_{31}\\
n-1& e=e_{33}.
\end{cases}
$$
Then define
\begin{eqnarray*}
\mathscr{E}_n&=&\{ (e,\tau_e(n),l)\ :\ e\in\mathcal E,\ 1\le l\le \tau_e(n)\}\\
\mathscr{E}_n^\ast&=&\{ (e,\tau_e(n),l)\in\mathscr{E}_n\ :\  e\ne e_{21},e_{23} \}.
\end{eqnarray*}

For any $w=(e,\tau_e(n),l)\in \mathscr{E}_n$, we use the notation $e_w:=e.$

For any $n,n^\prime\in\N$ and any
$(e,\tau_e(n),l)\in \mathscr{E}_n$ and $(e^\prime,\tau_e^\prime(n^\prime),
l^\prime)\in \mathscr{E}_{n^\prime}$ we say $(e,\tau_e(n),l)
(e^\prime,\tau_e^\prime(n^\prime),l^\prime)$ is {\it  admissible }
if the end point of $e$ is the initial point of $e^\prime.$
We denote it by $(e,\tau_e(n),l)\to(e^\prime,\tau_e^\prime(n^\prime),l^\prime)$.

 Define
$$
\Omega=\{\omega\in\mathscr{E}_{a_1}^\ast\times\prod_{k=2}^\infty
\mathscr{E}_{a_k}: \omega=\omega_1\omega_2\cdots \text{ s.t. }
\omega_{k}\to\omega_{k+1} \text{ for all } k\ge 1 \}.
$$
Define $\Omega_1=\mathscr{E}_{a_1}^\ast$ and for $n\ge 2,$ define
$$
\Omega_n=\{w\in\mathscr{E}_{a_1}^\ast\times\prod_{k=2}^n\mathscr{E}_{a_k}:
w=w_1\cdots w_n\text{ s.t. } w_{k}\to w_{k+1} \text{   for all } 1\le k<n \}.
$$
Define finally $\Omega_\ast=\bigcup_{n\ge 1}\Omega_n$.

Given any $w\in\Omega_n$, $1\le k<n$, we write $w=u*v$ or $w=uv$,
where $u=w_1\cdots w_k$, $v=w_{k+1}\cdots w_n$.

Given any $w\in\Omega_n$, define $B_w$ inductively as follows:
Let $B_{I}=[V-2,V+2]$ be the unique $(0,I)$-type band in
$\mathscr{G}_0$ and let $B_{III}=[-2,2]$ be the unique $(0,III)$-type band in $\mathscr{G}_0$.

Given $w\in\Omega_1$. If  $w=(e_{12},1,1),$
 then define $B_w$ to be the unique $(1,II)$-type band contained in  $B_I.$
 If $w=(e_{31},\tau_{e_{31}}(a_1),l)$, then  define $B_w$ to
 be the unique $l$-th $(1,I)$-type band contained in  $B_{III}.$
 If $w=(e_{33},\tau_{e_{33}}(a_1),l)$, then  define $B_w$
 to be the unique $l$-th $(1,III)$-type band contained in  $B_{III},$
 where we order the bands of the same type from left to right.

If $B_w$ has been defined for any $w\in \Omega_{n-1}$.
Given $w\in \Omega_n$ and write $w=w^\prime\ast(e,\tau_e(a_n),l)$,
then $w^\prime\in\Omega_{n-1}$. If $e=(T,T^\prime),$
define $B_w$ to be the unique $l$-th $(n,T^\prime)$-type band inside $B_{w^\prime}.$

With these notations we can rewrite \eqref{struc-spec} as
$$
\Sigma=\bigcap_{n\ge 0}\bigcup_{w\in\Omega_n} B_w.
$$

Given $w\in\Omega_k$, we say $w$ has length $k$ and denote by $|w|=k$.
If $B_w$ is of $(k,T)$ type, sometimes we also say simply that $B_w$ has type $T.$


\section{Variation, covariation and Gibbs like measure}\label{prepare}

In this section, we will present three  properties related to  the spectrum,
that is, bounded variation; bounded covariation and the existence of Gibbs like measures.
These properties play essential roles in the proof of the main theorems of the paper. We fix $V>0$ and $\alpha\in(0,1)$ irrational with continued fraction expansion $[0;a_1,a_2,\cdots]$.
Since now  $(a_k)_{k\ge 1}$ can be unbounded,
the proofs of these properties are much more difficult than \cite{FLW}, and we put the proofs to the sections \ref{bdvar}, \ref{bdcov} and \ref{gibbs-meas}.

We also collect several other basic properties which will be used in the
proofs of the main theorems.

\subsection{Bounded variation}\

\begin{theo}[Bounded variation]\label{bvar}
Let $V\ge20$ and $\alpha$ be irrational. Then there exists a constant
$\xi=\exp\left(180V\right)>1$  such that for any spectral
generating band $B$ of $\Sigma_{\alpha,V}$ with generating
polynomial $h$,
$$\xi^{-1}\le \left|\frac{h'(x_1)}{h'(x_2)}\right|\le \xi,\quad  \forall x_1,x_2\in B.$$

\end{theo}

\begin{cor}[Bounded distortion]\label{bdist}
Let $V\ge20$ and $\alpha$ be irrational. Then there exists a constant
$\xi=4\exp\left( 180V\right)>1$  such that for any spectral
generating band $B$
 of $\Sigma_{\alpha,V}$ with generating polynomial $h$,
$$\xi^{-1}\le|h'(x)|\cdot|B|\le \xi,\quad \forall x\in B.$$
\end{cor}

We will prove Theorem \ref{bvar} and Corollary \ref{bdist} in Section \ref{bdvar}.

\subsection{Bounded covariation}

\begin{theo}[Bounded covariation]\label{bco}
Let $V\ge 24$ and $\alpha$ be irrational. Then there exist absolute
constants $C_1,C_2>1$   such
that if $w,\widetilde w, wu,\widetilde wu\in\Omega_\ast$,
then, for $\eta=C_1\exp(2C_2V)$,
$$\eta^{-1} \frac{|{B}_{wu}|}{|{B}_{w}|}\le
\frac{|B_{\widetilde wu}|}{|B_{\widetilde w}|}\le \eta
\frac{|{B}_{wu}|}{|{B}_{w}|}.$$
\end{theo}

\begin{cor} \label{C-n}
Let $V\ge 24$ and $\alpha=[0;a_1,a_2,\cdots]$ be irrational.
Write $\mathscr{N}=\{n: n=a_i \text{ for some } i\}.$
Then there exist absolute constants $C_1,C_2>1$ and sequence
$\{\zeta_n:0<\zeta_n\leq 1, n\in\mathscr{N}\}$ depending on $\alpha, V$ and $n$,
 such that for any $k\in\mathbb{N}$   if $a_{k+1}=n, w\in\Omega_k, u=(e_{12},1,1)$ and
$wu\in\Omega_{k+1}$ then
$$
\eta^{-1}\zeta_n\leq \frac{|B_{wu}|}{|B_{w}|}\leq \eta
\zeta_n
$$
with $\eta=C_1\exp(C_2V).$   Moreover
$\zeta_1$ can be taken as $1.$
\end{cor}

 We will prove Theorem \ref{bco} and Corollary \ref{C-n} in Section \ref{bdcov}.

\subsection{Existence of Gibbs like measures}\

At first we introduce some notations used in this paper. For two positive
sequences $\{a_n\}$ and $\{b_n\}$, $a_n\sim b_n$ means that there
exist constants $0<d_1\leq d_2$ such that $d_1a_n\leq b_n \leq d_2
a_n $ for every $n\in \mathbb{N}.$ $a_n\lesssim b_n$ means  that
there exists a constant $d>0$ such that $a_n\leq d b_n$ for every
$n\in\mathbb{N}.$ $a_n\gtrsim b_n$ can  be defined similarly.

For any $\beta>0$ define
$$
b_{k,\beta}:=\sum_{w\in \Omega_k} |B_w|^\beta=\sum_{B\in\mathscr{G}_k}|B|^\beta.
$$

\begin{theo}[Existence of Gibbs like measures]\label{gibbs0}
For any $0<\beta<1$, there exists a probability measure $\mu_\beta$
supported on $\Sigma$ such that

 if $B_w$  has type $(k,I)$, let $u=(e_{12},1,1)$,  then
$$\mu_{\beta}(B_{w})\sim\frac{\zeta_{a_{k+1}}^\beta}{ a_{k+1}^{1-\beta}}
\frac{|B_{w}|^\beta}{b_{k,\beta}}\sim \frac{|B_{wu}|^\beta}{b_{k+1,\beta}}.
$$

If $B_w$ has type $(k,II)$, then
$$
\mu_{\beta}(B_{w})\sim
\frac{|B_{w}|^\beta}{b_{k,\beta}}.
$$

If $B_w$ has type $(k,III)$, then
$$
\mu_{\beta}(B_{w})\sim
\begin{cases}
\frac{|B_{w}|^\beta}{b_{k,\beta}}& a_{k+1}>1;\\
\frac{\zeta_{a_{k+2}}^\beta}{a_{k+2}^{1-\beta}}
\frac{|B_{w}|^\beta}{b_{k,\beta}}& a_{k+1}=1.
\end{cases}
$$
\end{theo}

We will prove  a generalized version of this
theorem, i.e. Theorem \ref{gibbs} in Section \ref{gibbs-meas}. Indeed the measure constructed in this theorem is a weak type Gibbs like measure, compared with that constructed in \cite{MRW,FLW}. However we still call it Gibbs like measure for convenience.

\subsection{Other useful facts}\

In this subsection we collect several other useful facts,
which are essentially contained in \cite{FLW}.

\begin{lem}{\rm (\cite{FLW})}\label{quo-deri}
Assume  $w\in\Omega_k, wu\in\Omega_{k+1}$ with $u=(e,p,l)$.
Let $h_w, h_{wu}$ be the  generating polynomials of $B_w, B_{wu}$ respectively.
Then for any $x\in B_{wu}$, if $e\ne e_{12}$,
$$\frac{V-8}{3}(p+1)\csc ^2\frac{l\pi}{p+1}
\le\left|\frac{h_{wu}^\prime(x)}{h_w^\prime(x)}\right|\le (V+5)(p+1)\csc ^2\frac{l\pi}{p+1},
$$
if $e=e_{12}$, then $p=1$, we have
$$\left(\frac{2(V-8)}{3}\right)^{a_{k+1}-1}
\le\left|\frac{h_{wu}^\prime(x)}{h_w^\prime(x)}\right|\le
\left(2(V+5)\right)^{a_{k+1}-1}.
$$
\end{lem}
We remark that here $p=\tau_e(a_{k+1})$.
This lemma is stated in another way in Proposition \ref{lm-2}.
See  \cite{FLW} Proposition 3.3 for a proof.

\begin{lem}\label{lem-bc}
Assume $V\ge 20.$ Write $t_1=(V-8)/3$ and $t_2=2(V+5).$
Then for any $w=w_1\cdots w_n \in\Omega_n$ with $w_i=(e_i,\tau_{e_i}(a_i),l_i)$ we have
\begin{equation}\label{bd-basic}
\prod_{e_i=e_{12}} \frac{1}{t_2^{a_i-1}}\cdot\prod_{e_i\ne e_{12}}
\frac{1}{t_2 a_i^3}\le|B_w|\le 4 \prod_{e_i=e_{12}} \frac{1}{t_1^{a_i-1}}
\cdot\prod_{e_i\ne e_{12}}\frac{1}{t_1 a_i}
\end{equation}
  Especially we have
  \begin{equation}\label{upper-bd}
|B_w|\le 4^{1-n/2}.
\end{equation}
\end{lem}

We will prove this lemma in the end of Section \ref{ladder}.


\section{Dimension formulas }\label{thm1}
 This section is devoted to the proof  of \eqref{formula-dim} in  Theorem \ref{fracdim}.
 At first we will show the box dimension formula, which is easier.
 Then we will propose a truncation  procedure to derive  the Hausdorff dimension formula. The formula \eqref{dim-box} will be proven in Section \ref{thm2}.

\subsection{$s^\ast(V)=\overline{\dim}_B\ \Sigma$}\

By \eqref{dimineq}, we only need to show that $\overline{\dim}_B\ \Sigma\le s^\ast(V).$
We recall an equivalent
definition of the upper box dimension(see for example \cite{Tr}).
Let $A\subset\mathbb{R}$ be a Cantor set. Let $a=\inf E$ and
$b=\sup E$. The complement of $A$, i.e. $[a,b]\backslash A$, consists
of countable many open intervals $\{G_i\}_{i\ge1}$, which is called
the gaps of $A$. For $i\ge1$, let $t_i=|G_i|$, then
$$\overline{\dim}_B\, A=\inf\left\{\beta\
 |\ \sum_{i=1}^\infty t_i^\beta<\infty\right\}.$$

We now consider all the gaps of the spectrum $\Sigma$.
  We call a gap  of {\it order $k$} if it is
covered by some band in  $\mathscr{G}_k$  but is not covered by any
 band in $\mathscr{G}_{k+1}$. Let $\mathscr{P}_k$ be the collection
of all gaps of order $k$, then
$$\overline{\dim}_B\, \Sigma=
\inf\left\{\beta\ |\ \sum_{k\ge0}\sum_{J\in\mathscr{P}_k}
|J|^\beta<\infty\right\}.$$

Now let us prove $\overline{\dim}_B\, \Sigma\le
s^\ast(V).$

If $s^\ast(V)=1,$ the result is trivial. So in the following we assume
$s^\ast(V)<1$. Fix $s$ such that $s^\ast(V)<s<1$.

Let $B$ be a generating band of order $k$. Suppose it contain $n$
generating bands of order $k+1$, then it contains $n-1$ gaps of
order $k$, which we denote  by $J_1,\cdots,J_{n-1}$. It is seen that
$|J_1|+\cdots+|J_{n-1}|\le |B|$. By concavity of the function $x^s$
we get
$$
\sum_{i=1}^{n-1}|J_i|^s=(n-1)\sum_{i=1}^{n-1}\frac{|J_i|^s}{n-1}
\le (n-1)\left(\frac{\sum_{i=1}^{n-1}|J_i|}{n-1}\right)^s \le
(n-1)^{1-s}|B|^s.
$$ For any generating band of order $k$, it
contain at most $2a_{k+1}+1$ generating bands of order $k+1$, so we
have
$$
\sum_{J\in\mathscr{P}_k}|J|^s\le (2a_{k+1})^{1-s}
\sum_{B\in\mathscr{G}_k}|B|^s =(2a_{k+1})^{1-s}b_{k,s}.
$$
By \eqref{b-k-b-k-1} and Lemma \ref{A-n-beta}(taking $\varepsilon$ to be 0), we have
$
b_{k,s}/b_{k+1,s}\sim a_{k+1}^{s-1}
$, so
there exists $C>0$ such that
$$\sum_{J\in\mathscr{P}_k}|J|^s\le Cb_{k+1,s}.$$

 Let $\varepsilon=s-s^*(V)$. Since $s^*(V)<s$, there exists $N>0$ such that
for any $k>N$, $s>s_k+\varepsilon/2$ and for any $B\in\mathscr{G}_k$ we
have $|B|<1$. By   \eqref{upper-bd} we have
$$
b_{k,s}\le \sum_{B\in\mathscr{G}_k}|B|^{s_k+\varepsilon/2} \le
4^{1-\varepsilon k/4}\sum_{B\in\mathscr{G}_k}|B|^{s_k} =4^{1-\varepsilon
k/4}.
$$
 Hence we have
$$\sum_{k\ge0}\sum_{J\in\mathscr{P}_k}|J|^s=\widetilde C+
\sum_{k=N}^\infty\sum_{J\in\mathscr{P}_k}|J|^s \le\widetilde C+
C\sum_{k=N+1}^\infty b_{k,s}<\infty.$$ So we get $\overline{\dim}_B\
\Sigma\le s$. Since $s>s^\ast(V)$ is arbitrary, we conclude that
$\overline{\dim}_B\
\Sigma\le s^\ast(V)$.

\subsection{$s_\ast(V)={\dim}_H\ \Sigma$}\

 Recall that $\alpha$ has expansion $[0;a_1,a_2,\cdots].$  If
$a_k$ is very  large, as discussed in the introduction,  the length of the bands of order
$k$ contained in the same band of order $k-1$ can differ from each other very much,
which makes the estimation very difficult. To overcome this difficulty,
  we propose the
following truncation  procedure.

Fix $0\le\varepsilon<1/12$.
Define
\begin{equation}\label{trunc}
\mathscr{E}_n(\varepsilon)=\{(e_{12},1,1)\}\bigcup \bigcup_{e\ne e_{12}}\{(e,\tau_e(n),l) :
(\tau_e(n)+1)\varepsilon< l < (\tau_e(n)+1)(1-\varepsilon)\}.
\end{equation}
It is seen that if $n\le 10,$ then $\mathscr{E}_n(\varepsilon)=\mathscr{E}_n.$ Define
\begin{equation}\label{truncO}
\Omega_n(\varepsilon)= \Omega_n\bigcap \left(\mathscr{E}_{a_1}^\ast
\times\prod_{k=2}^n\mathscr{E}_{a_k}(\varepsilon)\right)\ (n\ge 2),\quad
\Omega_\ast(\varepsilon)=\bigcup_{n\ge 2}\Omega_n(\varepsilon)
\end{equation}
and
\begin{equation}\label{truncE}
E_\varepsilon:=\bigcap_{n\geq 1}\bigcup_{w\in
\Omega_{n}(\varepsilon)}B_w.
\end{equation}
It is obvious that $E_\varepsilon\subset \Sigma.$ For this set we can also
define the associated pre-dimensions $s_\ast(\varepsilon)$.

\begin{prop}\label{dim-converge}
$s_\ast(\varepsilon)\to s_\ast$ when $\varepsilon\to 0.$
\end{prop}

\begin{proof} We begin with the comparison of $s_n$ and $s_n(\varepsilon).$
By the definitions
$$
\sum_{w\in \Omega_{n}}|B_w|^{s_n}=1;\ \ \
\sum_{w\in \Omega_{n}(\varepsilon)}|B_w|^{s_n(\varepsilon)}=1.
$$
It is seen that $s_n(\varepsilon)\leq s_n$ for $n\in\mathbb{N}$.

\noindent{\bf Claim:} There exists a constant $C>0$ such that, for
any small $\varepsilon$,
\begin{equation}\label{preclaim}
\sum_{w\in \Omega_{n}(\varepsilon)}|B_w|^{s_n}\geq
(1-C\varepsilon)^n.
\end{equation}
We first show that the claim implies the result. In fact
 if the claim holds, then
 \begin{eqnarray*}
 (1-C\varepsilon)^n&\leq&\sum_{w\in
 \Omega_{n}(\varepsilon)}|B_w|^{s_n}\\
 &=&\sum_{\omega\in \Omega_{n}(\varepsilon)}|B_w|^{s_n(\varepsilon)}
|B_w|^{s_n-s_n(\varepsilon)}\\
&\leq&4\cdot 4^{-(s_n-s_n(\varepsilon))n/2}\sum_{w\in
\Omega_{n}(\varepsilon)}|B_w|^{s_n(\varepsilon)}\\
&=&4\cdot 4^{-(s_n-s_n(\varepsilon))n/2},
\end{eqnarray*}
where the second inequality is due to \eqref{upper-bd}.
Consequently
$$
0\leq s_n-s_n(\varepsilon)\leq \frac{2}{n}-\frac{\ln(1-C\varepsilon)}{\ln
2}.
$$
From this we can conclude that $s_\ast(\varepsilon)\to s_\ast$ when $\varepsilon\to0.$

Now we go back to the proof of  the claim. For this purpose we introduce the
following intermediate symbolic spaces. For $1\leq j\leq n-1$ define
$$
\Omega_{n}^{(j)}(\varepsilon):= \Omega_n\cap
\left(\mathscr{E}_{a_1}^\ast\times\prod_{l=2}^j
\mathscr{E}_{a_l}\times \prod_{l=j+1}^n \mathscr{E}_{a_l}(\varepsilon)\right).
$$
Thus $\Omega_{n}^{(1)}(\varepsilon)=\Omega_{n}(\varepsilon).$ We also
write $\Omega_{n}^{(n)}(\varepsilon):=\Omega_{n}$ to unify the
notation.

To prove \eqref{preclaim}, we only need to show that, for
$j=1,\cdots,n-1$,
\begin{equation}\label{omega-j-m}
\sum_{w\in
\Omega_{n}^{(j)}(\varepsilon)}|B_w|^{s_n}\geq(1-C\varepsilon)
\sum_{w\in \Omega_{n}^{(j+1)}(\varepsilon)}|B_w|^{s_n}.
\end{equation}

By the definition if $a_{j+1}\leq 10,$ then
$\mathscr{E}_{a_{j+1}}(\varepsilon)=\mathscr{E}_{a_{j+1}}$ and consequently
$\Omega_{n}^{(j)}(\varepsilon)=\Omega_{n}^{(j+1)}(\varepsilon)$, thus
\eqref{omega-j-m} is trivial.

Now assume $a_{j+1}>10.$  We can write
$$
\displaystyle\sum_{w\in \Omega_{n}^{(j+1)}(\varepsilon)}
|B_w|^{s_n}=\sum_{e\in\mathcal E} Z_e\ \ \ \text{ and }\ \ \
\displaystyle\sum_{w\in
\Omega_{n}^{(j)}(\varepsilon)}|B_w|^{s_n}=\sum_{e\in\mathcal E} Z_e(\varepsilon),
 $$
where for $e=e_{12}$
$$
Z_{e_{12}}=Z_{e_{12}}(\varepsilon)=\displaystyle \sum_{w\in
\Omega_{n}^{(j+1)}(\varepsilon),w_{j+1}=(e_{12},1,1)}
|B_w|^{s_n}
$$
and for $e\ne e_{12}$
$$
\begin{array}{lcll}
Z_e&=&\displaystyle\sum_{l=1}^{\tau_e(a_{j+1})}
\sum_{w\in \Omega_{n}^{(j+1)}(\varepsilon),w_{j+1}=(e,\tau_e(a_{j+1}),l)}
|B_w|^{s_n},\\
Z_e(\varepsilon)&=&\displaystyle\sum_{l=[a_{j+1}\varepsilon]}^{[a_{j+1}(1-\varepsilon)]}
\sum_{w\in \Omega_{n}^{(j+1)}(\varepsilon),w_{j+1}=(e,\tau_e(a_{j+1}),l)}
|B_w|^{s_n}
\end{array}
$$

Note that for $e\ne e_{12}$, $Z_e(\varepsilon)$ is different from $Z_e$
only in the range of the index $l$. To show \eqref{omega-j-m} it is
enough to show that
$$
Z_e(\varepsilon)\geq (1-C\varepsilon)Z_e, \ \ \ e\in\mathcal E.
$$

The case $e=e_{12}$ is trivial. Now we consider $e=e_{21}.$
 In this case we have $\tau_e(a_{j+1})=a_{j+1}+1.$Write $\theta_l=(e,\tau_e(a_{j+1}),l)$.
 Then we can rewirte $Z_e(\varepsilon)$ as
$$
Z_e(\varepsilon)=\displaystyle\sum_{l=[a_{j+1}\varepsilon]}^{[a_{j+1}(1-\varepsilon)]}
\sum_{u\theta_lv\in \Omega_{n}^{(j+1)}(\varepsilon) }
|B_w|^{s_n}
$$

 Fix
$u\in\Omega_{j}$, write
$$\displaystyle\gamma_l:=\sum_{
u\theta_lv\in\Omega_{n}^{(j+1)}(\varepsilon)}|B_{u\theta_lv}|^{s_n}.
$$
Let $l_0=[a_{j+1}/2].$ By Theorem \ref{bco} we
have
$$
\gamma_l=|B_{u\theta_l}|^{s_n}\sum_{
u\theta_lv\in\Omega_{n}^{(j+1)}(\varepsilon)}
\frac{|B_{u\theta_l
v}|^{s_n}}{|B_{u\theta_l}|^{s_n}}\ \ \sim\ \
\frac{|B_{u\theta_l}|^{s_n}}{|B_{u\theta_{l_0}}|^{s_n}}
\gamma_{l_0}.
$$
By Corollary \ref{bdist} and   Lemma \ref{quo-deri},
$$
\frac{|B_{u\theta_l}|}{|B_{u\theta_{l_0}}|}=
\frac{|B_{u\theta_l}|/|B_{u}|}{|B_{u\theta_{l_0}}|/|B_{u}|}
\sim \sin^2\frac{l\pi}{a_{j+1}+1}.$$
Then there exists $c>1$
independent to $l,j$ such that
$$
c^{-1}\gamma_{l_0}\sin^{2s_n}\frac{l\pi}{a_{j+1}+1}\le
\gamma_l\le c\gamma_{l_0}\sin^{2s_n}\frac{l\pi}{a_{j+1}+1}.
$$ So we
have
\begin{eqnarray*}
\sum_{l=1}^{a_{j+1}+1}\gamma_l&\ge&
c^{-1}\gamma_{l_0}\sum_{l=1}^{a_{j+1}+1}
\sin^{2s_n}\frac{l\pi}{a_{j+1}+1}\\
&\ge& c^{-1}\gamma_{l_0}\sum_{l=l_0/2}^{3l_0/2}
\sin^{2s_n}\frac{l\pi}{a_{j+1}+1}\\
&\ge&c^{-1}\gamma_{l_0}l_0\sin^{2s_n}\pi/4.
\end{eqnarray*}
Let $A=\{1\le
l\le a_{j+1}+1\ |\ l<[a_{j+1}\varepsilon]$ or $l>[a_{j+1}(1-\varepsilon)]\}$,
we have
$$
\sum_{l\in A}\gamma_l\le c\gamma_{l_0}\sum_{l\in A}
\sin^{2s_m}\frac{l\pi}{a_{j+1}+1}\le c\gamma_{l_0}\cdot
2a_{j+1}\varepsilon\cdot\sin^{2s_n}\pi/4.
$$
This implies for all $u\in\Omega_{j}$,
$$
\sum_{l=[a_{j+1}\varepsilon]}^{[a_{j+1}(1-\varepsilon)]}\gamma_l\ge
(1-4c^2\varepsilon)\sum_{l=1}^{a_{j+1}+1}\gamma_l,
$$
so we get $Z_e(\varepsilon)\geq (1-4c^2\varepsilon)Z_e.$

For other $e\ne e_{12}$, the proof is the same. Thus the proof is completed.
\end{proof}

\begin{prop}\label{dim-epsilon}
$\dim_H E_\varepsilon=s_*(\varepsilon)$.
\end{prop}
\begin{proof}
It is known  that $\dim_H E_\varepsilon\leq s_*(\varepsilon)$, so it
only need to show $\dim_H E_\varepsilon\ge s_*(\varepsilon).$ It is trivial if $s_\ast(\varepsilon)=0$, we thus assume
$s_\ast(\varepsilon)>0.$
Fix any $0<\beta<s_*(\varepsilon)$ and define
\begin{equation}\label{b-k-beta-epsilon}
b_{k,\beta}(\varepsilon):=\sum_{\omega\in\Omega_{k}(\varepsilon)}|B_\omega|^\beta.
\end{equation}
Then there exists $K\in\N$ such that $b_{k,\beta}(\varepsilon)\ge 1$ for any $k\ge K.$

By \eqref{b-k-b-k-1} and Lemma \ref{A-n-beta}, we have
\begin{equation}\label{b-k-b-k-1-1}
\frac{b_{k-1,\beta}(\varepsilon)}{b_{k,\beta}(\varepsilon)}\sim
a_k^{\beta-1}.
\end{equation}
By Theorem \ref{gibbs}, we can
 construct a Gibbs like measure $\mu_{\beta,\varepsilon}$
supported on $E_\varepsilon$.
Define $\delta_0:= \min\{|B_w|: w\in \Omega_{K+1}(\varepsilon)\}$.

\noindent {\bf Claim: }
 for any  open
interval $U\subset \mathbb{R}$ with $|U|\le \delta_0/2$ we have
$$\mu_{\beta,\varepsilon}(U)\lesssim |U|^\beta.$$

\noindent $\lhd$
Take any  open interval $U\subset\mathbb{R}$ with $|U|\le \delta_0/2$, define
$$\Xi:=\{w\in \Omega_\ast(\varepsilon):B_w\cap U\ne \emptyset,\
 |B_w|\le |U|< |B_{w^-}|\},$$
where $w^-$ is gotten  by deleting  the last symbol of
$w$. At first we claim that $|w|\ge K+1$ for any $w\in\Xi.$
In fact if otherwise, there exists $\tilde w=wu\in \Omega_{K+1}(\varepsilon)$. Then
$$
\delta_0\le |B_{\tilde w}|\le |B_{w}|\le |U|\le \delta_0/2,
$$
which is a contradiction.

Notice that by the natural covering property,
any two generating bands are either disjoint or one of them is included in another,
thus we conclude that  $\#\{w^-:w\in\Xi\}\leq 2.$ Thus to show the claim we
only need to show that
$$\mu_{\beta,\varepsilon}(U\cap B_{w^{-}})\lesssim |U|^{\beta},
 \ \ (\forall w\in \Xi).$$

If $B_{w^{-}}$ is of type $(k,I)$, then by \eqref{gibbs1},
$$
\mu_{\beta,\varepsilon}(U\cap B_{w^{-}})\le
\mu_{\beta,\varepsilon}(B_{w^{-}}) \sim
\frac{|B_w|^{\beta}}{b_{k+1,\beta}(\varepsilon)}\le |U|^{\beta}.
$$

Now we assume $B_{w^{-}}$ is a band of type $(k,II)$ or
$(k,III)$. In this case, $w$ has the form $w^-
(e,p,l)$ with $e\ne e_{12}$ and $p=\tau_e(a_{k+1}).$
By bounded variation and
 Lemma \ref{quo-deri}, if $a_{k+1}\le 10$, then
 $$
 \frac{|B_w|}{|B_{w^-}|}\sim \frac{1}{a_{k+1}}.
 $$
If $a_{k+1}>10$, then we have $\varepsilon\leq
l/a_{k+1}\leq 1-\varepsilon.$ Consequently, also by bounded variation and Lemma \ref{quo-deri},
there exists constants $C>1$ not depending on $\varepsilon$ such that
$$
\frac{\varepsilon^2}{C a_{k+1}}\le \frac{\sin^2\varepsilon\pi}{C a_{k+1}}\le
\frac{|B_\omega|}{|B_{\omega^-}|}\le \frac{C}{a_{k+1}}.
$$
So in both cases we have
\begin{equation}\label{two-interval}
\frac{|B_w|}{|B_{w^-}|}\sim \frac{1}{a_{k+1}}.
\end{equation}

  Let
$$
\Upsilon:=\{u: |u|=k+1; B_u\subset B_{w^-};
B_u\cap U\ne \emptyset\}.
$$
Then by \eqref{two-interval}, \eqref{gibbs2} and \eqref{gibbs3}
\begin{eqnarray*}
\mu_{\beta,\varepsilon}(U\cap
B_{w^-})&=&\sum_{u\in\Upsilon}
\mu_{\beta,\varepsilon}(B_u)\leq
\#\Upsilon\cdot\max_{u\in\Upsilon}
\mu_{\beta,\varepsilon}(B_u)\\
&\lesssim&\frac{|U|a_{k+1}}{
|B_{w^-}|}\cdot
\frac{(|B_{w^-}|/a_{k+1})^\beta}{b_{k+1,\beta}(\varepsilon)}\\
&=&\frac{|U| a_{k+1}^{1-\beta}}{|B_{w^-}|^{1-\beta}
\  b_{k+1,\beta}(\varepsilon)} =:d_1.
\end{eqnarray*}

On the other hand, we trivially have
$$
\mu_{\beta,\varepsilon}(U\cap B_{\omega^{-}})\le
\mu_{\beta,\varepsilon}(B_{\omega^{-}})\lesssim
\frac{|B_{\omega^-}|^\beta}{b_{k,\beta}(\varepsilon)}=:d_2.
$$
So, we have
$$\begin{array}{rcl}
\mu_{\beta,\varepsilon}(U\cap B_{\omega^{-}})
&\lesssim& \min\{d_1,d_2\}\leq d_1^\beta d_2^{1-\beta}\\
 &=&|U|^\beta
\left(\frac{a_{k+1}^{(1-\beta)}b_{k,\beta}(\varepsilon)}
{b_{k+1,\beta}(\varepsilon)}\right)^\beta \frac{1}{b_{k,\beta}(\varepsilon)}\\
&\lesssim&|U|^\beta
\end{array}
$$
where the last inequality is due to \eqref{b-k-b-k-1-1} and
$b_{k,\beta}(\varepsilon)\geq 1$. \hfill $\rhd$

Then by the mass distribution principle we conclude that $\dim_H
E_\varepsilon\geq \beta.$ Since $\beta<s_\ast(\varepsilon)$ is arbitrary,
we get $\dim_H E_\varepsilon\geq s_\ast(\varepsilon).$
\end{proof}

\noindent {\bf Proof of $\dim_H\, \Sigma =s_\ast(V)$.}\
By \eqref{dimineq}, we only need to prove
$$\dim_H\, \Sigma\ge s_\ast(V) .$$
  Since $E_\varepsilon\subset
\Sigma,$ by Proposition \ref{dim-epsilon} we have
$$
\dim_H \Sigma\geq \dim_H E_\varepsilon= s_\ast(\varepsilon).
$$
By Proposition \ref{dim-converge} we get $\dim_H
\Sigma\geq s_\ast(V).$
\hfill $\Box$

\section{Proof of \eqref{dim-box} and  Theorem   \ref{limfor}}\label{thm2}

The main result of this section is Proposition \ref{upper-lower}. Formula \eqref{dim-box} and Theorem \ref{limfor} are  direct consequences  of this proposition.

 Let us do some preparation. Recall that we have defined $K_\ast(\alpha)$ and $K^\ast(\alpha)$ in \eqref{K-u-l}. To simplify the notation, we write $K_\ast=K_\ast(\alpha)$ and  $K^\ast=K^\ast(\alpha)$.
It is obvious that $1\le K_\ast\le K^\ast.$

Throughout this section, for a matrix $A\in {\rm M}(3,\R)$ we define
$$
\|A\|:=\max\{|a_{ij}|:1\leq i,j\leq 3\}.
$$
For any $n\ge1$, $0\le x\le1$ define
\begin{equation}\label{rnb}
\begin{array}{l}
{\mathbf R}_n(x) := \begin{pmatrix}
0&x^{(a_n-1)}&0\\
(a_n+1)x&0&a_n x\\
a_n x&0&(a_n-1)x
\end{pmatrix},\\
\mathbf{S}_n(x) :={\mathbf R}_1(x)\cdots\mathbf{R}_n(x).
\end{array}
\end{equation}
For $x\in[0,1]$ we define
$$
\psi(x):=\liminf_{n\to\infty}\|
\mathbf{S}_n(x)\|^{1/n},\ \ \
\phi(x):=\limsup_{n\to\infty}\|
\mathbf{S}_n(x)\|^{1/n}.
$$
Then  $\psi(x)\leq \phi(x)$ for
$x\in[0,1].$
It is direct to get the following inequality,
assume $0<x<y\leq 1$, then for any $k>0$,
\begin{equation}\label{2-2}
\mathbf{R}_{k}(x)\mathbf{R}_{k+1}(x)\leq \frac{x}{y}\cdot
\mathbf{R}_{k}(y)\mathbf{R}_{k+1}(y).
\end{equation}

\begin{lem}\label{infsup}
(1) \ If $ K_\ast<\infty,$ then  $\psi:[0,1]\to \mathbb{R}^+$ is
strictly increasing  and
$$
\left(\frac{K_\ast}{2}\vee\sqrt[3]{2}\right)x\leq\psi(x)\leq K_\ast
\sqrt{2x}, \quad \forall x\in[0,1].
$$

If $ K_\ast=\infty$, then $\psi(0)=0$ and $\psi(x)=\infty$ for any
$x\in(0,1].$

(2) \ If $ K^\ast<\infty,$ then  $\phi:[0,1]\to \mathbb{R}^+$ is
strictly increasing and
$$
\left(\frac{K^\ast}{2}\vee\sqrt[3]{2}\right) x\leq\phi(x)\leq
K^\ast\sqrt{2x},\quad \forall x\in[0,1].
$$

If $ K^\ast=\infty$, then $\phi(0)=0$ and $\phi(x)=\infty$ for any
$x\in(0,1].$
\end{lem}

\begin{proof} We only prove (1), since the proof of (2) is analogous.
It is easy to check that $\psi(0)=0.$ Define
\begin{equation}\label{allone}
 J=\left(
\begin{array}{ccc}
1&1&1\\
1&1&1\\
1&1&1
\end{array}
\right).
\end{equation}

At first assume $K_\ast<\infty.$ It is ready to check the following
inequality
$$
\mathbf{R}_n(x)\mathbf{R}_{n+1}(x)\leq 2xa_na_{n+1}J.
$$
Consequently we have
$$\|\mathbf{R}_n(x)\mathbf{R}_{n+1}(x)\|\le 2xa_na_{n+1},$$
which implies $\psi(x)\le K_*\sqrt{2x}$.

\noindent{\bf Claim:} There exists a path $i_0i_1\cdots
i_n\in\{1,2,3\}^{n+1}$ such that for $k=1,\cdots, n$,
$$
\left(\mathbf{R}_k(x)\right)_{i_{k-1}i_k}\geq (1\vee \frac{a_k}{2})
x,
$$  and for $k=1,\cdots, n-2$,
$$
\left(\mathbf{R}_k(x)\right)_{i_{k-1}i_k}
\left(\mathbf{R}_{k+1}(x)\right)_{i_{k}i_{k+1}}
\left(\mathbf{R}_{k+2}(x)\right)_{i_{k+1}i_{k+2}}\geq 2x^3.
$$

\noindent $\lhd$ Take a path $i_0i_1\cdots i_n\in\{1,2,3\}^{n+1}$ as
follows:
\begin{itemize}
\item take $i_0=2$;
\item for $0<j\le n$, if $i_{j-1}=1$, then take $i_j=2$;
\item for $0<j\le n$, if $i_{j-1}=2$ or $3$ and $a_{j+1}>2$, then take $i_j=3$;
\item for $0<j\le n$, if $i_{j-1}=2$ or $3$ and $a_{j+1}\le2$, then take $i_j=1$.
\end{itemize}

Fix $1\leq k\leq n$ let us discuss the following cases:

\noindent {\bf Case 1}: $i_{k-1}i_k=12.$ By the way we choose the
path, we have $a_{k}\leq 2$. Thus
$$
(\mathbf{R}_k(x))_{i_{k-1}i_k}=(\mathbf{R}_k(x))_{12}=x^{a_k-1}\geq
(1\vee \frac{a_k}{2})x.
$$

\noindent {\bf Case 2}: $i_{k-1}i_k=21.$  Thus
$$
(\mathbf{R}_k(x))_{i_{k-1}i_k}=(\mathbf{R}_k(x))_{21}=(a_k+1)x\geq
(2\vee a_k)x.
$$

\noindent {\bf Case 3}: $i_{k-1}i_k=23.$  Thus
$$
(\mathbf{R}_k(x))_{i_{k-1}i_k}=(\mathbf{R}_k(x))_{23}=a_kx\geq x.
$$

\noindent {\bf Case 4}: $i_{k-1}i_k=31$ or $33.$   By the way we
choose the path, we have $a_k>2$. Thus
$$
(\mathbf{R}_{k}(x))_{i_{k-1}i_{k}}=
\begin{cases}
a_{k}x& i_{k}=1\\
(a_{k}-1)x& i_{k}=3
\end{cases}
 \geq (2\vee \frac{a_k}{2})x.
$$

Since all the possible sequences of $i_{k-1}i_ki_{k+1}i_{k+2}$ are
\begin{eqnarray*}
&&\{1212; 1231;1233; 2121;2123; 2312;2331;2333;\\
&&\ \ 3121;3123; 3312;3331;3333\},
\end{eqnarray*}
by the conclusions of four cases above we get the result. \hfill $\rhd$

Thus by the claim above, on the one hand we have
$$
\|\mathbf{R}_1(x)\cdots \mathbf{R}_n(x)\| \ge (\mathbf{R}_1(x))_{i_0
i_1}\cdots (\mathbf{R}_n(x))_{i_{n-1} i_n} \ge 2^{[n/3]}x^n,
$$
which implies $\psi(x)\ge \sqrt[3]{2}\ x$. On the other hand we get
\begin{equation}\label{rn}
\|\mathbf{R}_1(x)\cdots \mathbf{R}_n(x)\| \ge (\mathbf{R}_1(x))_{i_0
i_1}\cdots (\mathbf{R}_n(x))_{i_{n-1} i_n} \ge a_1\cdots
a_n2^{-n}x^n,
\end{equation}
which implies $\psi(x)\ge K_\ast x/2$.

Now we are going to show that $\psi$ is strictly increasing.
By the definition of $\psi$ and \eqref{2-2} we get
$$\psi(x)\leq (\frac{x}{y})^{1/2}\psi(y).$$
Since $\psi(x)\geq \sqrt[3]{2}\ x$ we have  $\psi(x)> 0$ when $x>0$,
thus we conclude that $\psi$ is strictly increasing.

If $K_\ast=\infty,$  by \eqref{rn} we have $\psi(x)=\infty$ for any
$x\in(0,1]$.
\end{proof}

Due to the strictly increasing property of $\psi$ and $\phi$ we
can define
\begin{equation}\label{f-u-l-alpha}
\begin{cases}
f_\ast(\alpha)&:=\inf\{x>0: \psi(x)\geq 1\}=\sup\{x\geq 0:\psi(x)\leq 1\};\\
f^\ast(\alpha)&:=\inf\{x>0: \phi(x)\geq 1\}=\sup\{x\geq 0:\phi(x)
\leq 1\}.
\end{cases}
\end{equation}

\begin{cor}\label{K-ast}
If $K_\ast<\infty$, then
$$
f_\ast(\alpha)\in\left[\frac{1}{2K_\ast^2},\frac{2}{K_\ast}\wedge
\frac{1}{\sqrt[3]{2}}\right]\subset(0,1);
$$ if
$K_\ast=\infty$, then $f_\ast(\alpha)=0$.

If $K^\ast<\infty$, then
$$
f^\ast(\alpha)\in\left[\frac{1}{2K^{\ast2}},\frac{2}{K^\ast}\wedge
\frac{1}{\sqrt[3]{2}}\right]\subset(0,1);
$$
 if $K^\ast=\infty$, then
$f^\ast(\alpha)=0$.
\end{cor}

Now we state the main result of this section:

\begin{prop}\label{upper-lower}
Let $V\ge 24$ and $\alpha=[0;a_1,a_2,\cdots]$ be irrational.

(i)\ If  $K_\ast<\infty$, then $0<f_\ast(\alpha)<1$ and
\begin{equation}\label{K-low-ast}
\frac{-\ln f_\ast(\alpha)}{6\ln 4K_\ast^2+\ln 2(V+5)}\le S_\ast(V)\le
\frac{-\ln f_\ast(\alpha)}{\ln (V-8)/3}.
\end{equation}

If $K_\ast=\infty$, then $f_\ast(\alpha)=0$ and  $s_\ast(V)=1.$

(ii)\ If  $K^\ast<\infty$, then $0<f^\ast(\alpha)<1$ and
\begin{equation}\label{K-up-ast}
\frac{-\ln f^\ast(\alpha)}{6\ln 2K^\ast+\ln 2(V+5)}\le S^\ast(V)\le
\frac{-\ln f^\ast(\alpha)}{\ln (V-8)/3}\wedge \frac{\ln K^*+\ln\sqrt{2}}{\ln K^*+\ln (V-8)/3} .
\end{equation}

If $K^\ast=\infty$, then $f^\ast(\alpha)=0$ and $s^\ast(V)=1.$
\end{prop}

\begin{proof}
Let $t_1=(V-8)/3$, $t_2=2(V+5)$.
Write
$$\bar{a}_k=a_k+1,\quad \underline{a}_k=a_k-1,\quad
\delta_k=(a_1\cdots a_k)^{1/k}.$$
For any $0\le\gamma\le1$ define
$$\begin{array}{l}
Q_k(\gamma):= \left(
\begin{array}{ccc}
0&t_1^{-\gamma \underline{a}_k}&0\\
\bar{a}_k(t_1a_k)^{-\gamma}&0&a_k(t_1a_k)^{-\gamma}\\
a_k(t_1a_k)^{-\gamma}&0&\underline{a}_k(t_1a_k)^{-\gamma}
\end{array}
\right),\\
\tilde Q_k(\gamma):= \left(
\begin{array}{ccc}
0&t_2^{-\gamma \underline{a}_k}&0\\
\bar{a}_k(t_2a_k^3)^{-\gamma}&0&a_k(t_2a_k^3)^{-\gamma}\\
a_k(t_2a_k^3)^{-\gamma}&0&\underline{a}_k(t_2a_k^3)^{-\gamma}
\end{array}
\right).
\end{array}
$$

By \eqref{bd-basic} we have
\begin{equation}\label{bulb}
\begin{array}{l}
b_{k,\gamma}\le 4^\gamma
(1,0,1)Q_1(\gamma)\cdots Q_k(\gamma)(1,1,1)^t,\\
b_{k,\gamma}\ge(1,0,1)\tilde Q_1(\gamma)
\cdots \tilde Q_k(\gamma)(1,1,1)^t,
\end{array}\end{equation}
where $(1,1,1)^t$
is a column vector.
We have, by definition of norm and \eqref{bulb},
\begin{equation}\label{ub-box}
b_{k,\gamma}\le 6\cdot4^\gamma
||Q_1(\gamma)\cdots Q_k(\gamma)||.\end{equation}
By \eqref{rnb}, for any $k\ge1$, $Q_k(\gamma)\leq \mathbf{R}_k(t_1^{-\gamma})$,
then by \eqref{bulb},
\begin{equation}\label{ub2}
b_{k,\gamma}\le4^\gamma
(1,0,1)  \mathbf{S}_k(t_1^{-\gamma})(1,1,1)^t\le 6\cdot4^\gamma\|\mathbf{S}_k(t_1^{-\gamma})\|.
\end{equation}
Since for any $k\ge1$, $\tilde Q_k(\gamma)\ge
a_k^{-3\gamma}\mathbf{R}_k(t_2^{-\gamma})$, by \eqref{bulb},
$$b_{k,\gamma}\ge \delta_k^{-3k\gamma}  (1,0,1) \mathbf{S}_k(t_2^{-\gamma})(1,1,1)^t.$$
Since there exists $c>1$ such that $c^{-1}J<\mathbf{S}_5(t_2^{-\gamma})<cJ$
(see \eqref{allone} for definition of $J$)
and by definition of norm,
$$\begin{array}{rcl}
b_{k,\gamma}&\ge& c^{-1}\delta_k^{-3k\gamma}
\|\mathbf{R}_6(t_2^{-\gamma})\mathbf{R}_7(t_2^{-\gamma})\cdots \mathbf{R}_k(t_2^{-\gamma}) \|\\
&\ge&c^{-1}\delta_k^{-3k\gamma}\|\mathbf{S}_k(t_2^{-\gamma})\|/
\|\mathbf{S}_5(t_2^{-\gamma})\|\\
&\ge&c^{-2}\delta_k^{-3k\gamma}\|\mathbf{S}_k(t_2^{-\gamma})\|.
\end{array}$$
This implies, on one hand, by Claim in Lemma \ref{infsup},
\begin{equation}\label{lbk1}
b_{k,\gamma}
 \geq c^{-2}\delta_k^{-3k\gamma}\frac{a_1\cdots a_k}{2^k} t_2^{-k\gamma}
= c^{-2}\left((\delta_k^3
t_2)^{-\gamma}\delta_k/2\right)^k;
\end{equation}
on the other hand, by \eqref{2-2},
\begin{equation}\label{lbk2}
b_{k,\gamma}\geq c^{-2}\|
\mathbf{S}_{k}((\delta_{k}^6t_2)^{-\gamma})\|.
\end{equation}

We discuss first upper bound of pre-dimensions.

Assume  $K_*<\infty$. Then $f_\ast(\alpha)>0$ by Corollary \ref{K-ast}.
Take $\gamma>0 $ such that
$t_1^{-\gamma}<f_*(\alpha)$, i.e. $\gamma>-\ln
f_*(\alpha)/\ln t_1$. Then by the definition of $f_\ast(\alpha)$
and the fact that $\psi$ is strictly increasing on $[0,1]$
we conclude that  $\psi(t_1^{-\gamma})<1$. Thus for any
$\lambda\in (\psi(t_1^{-\gamma}),1)$ and any $n$ big enough, there exists
$k_n\ge n$ such that
$$
\|  \mathbf{S}_{k_n}(t_1^{-\gamma})
\|<\lambda^{k_n}.
$$
Thus  $b_{{k_n},\gamma}<1$ when $n$ is big enough  by \eqref{ub2}.
Consequently $\gamma>s_{k_n}$ for $n$ big.
Thus we conclude that  $s_\ast(V)\leq \gamma.$ Since
$\gamma>{-\ln f_*(\alpha)}/{\ln t_1}$ is arbitrary, we get
$$
s_*(V)\le\frac{-\ln f_*(\alpha)}{\ln t_1}=\frac{-\ln
f_*(\alpha)}{\ln (V-8)/3},
$$
which is the second inequality in \eqref{K-low-ast}.

Assume $K^*<\infty$. By essentially repeating  the above proof
(indeed it is simpler), we get
$$
s^*(V)\le\frac{-\ln f^*(\alpha)}{\ln (V-8)/3},
$$
which is one of the second inequality in \eqref{K-up-ast}.

On the other hand, by definition of the norm, it is direct to check  that
$$||Q_k(\gamma)Q_{k+1}(\gamma)||\le 2(a_ka_{k+1})^{1-\gamma}t_1^{-\gamma}.$$
(Note that if $a_k=1$, then $t_1^{-\gamma \underline{a}_k}=1$,
we cannot get $||Q_k(\gamma)||\le 2 a_k^{1-\gamma}t_1^{-\gamma}$ in this case.)
So, by \eqref{ub-box}, for some $c>0$,
$$b_{k,\gamma}\le 6c\cdot4^\gamma (2t_1^{-\gamma})^{k/2}(a_1\cdots a_k)^{1-\gamma}
\le 24c\left(\frac{\sqrt{2}\delta_k}
{(\delta_k t_1)^{\gamma}}\right)^k
.$$
This implies
$$s^*(V)\le\frac{\ln K^*+\ln\sqrt{2}}{\ln K^*+\ln t_1},$$
and we get the second inequality in \eqref{K-up-ast}.

Next we discuss the low bound of pre-dimensions.

Assume $K_\ast<\infty$. We will show the first inequality in \eqref{K-low-ast}.
By Corollary \ref{K-ast} we have
$0<f_\ast(\alpha)<1$. Fix $g\in (f_\ast(\alpha),1)$.

\medskip

\noindent {\bf  Claim 1.}
 If $\gamma<\frac{-\ln g}{{6\ln 4K_\ast^2+\ln t_2}},$
 then $\lim_{k\rightarrow\infty}b_{k,\gamma}=\infty.$

\medskip

\noindent $\vartriangleleft$
Notice that by Corollary \ref{K-ast} we have
$$
g>f_\ast(\alpha)\geq \frac{1}{2K_\ast^2}.
$$

If $\delta_k\ge 4K_\ast^2,$ then
$$\gamma<\frac{\ln 2K_\ast^2}{{6\ln 4K_\ast^2+\ln t_2}}<
\frac{\ln 4K_\ast^2-\ln2}{3\ln 4K_\ast^2+\ln t_2}\le
\frac{\ln \delta_k-\ln2}{3\ln \delta_k+\ln t_2}.$$
  Write
$$
c:=\frac{\ln 4K_\ast^2-\ln2}{3\ln 4K_\ast^2+\ln t_2}-\gamma>0,
$$
then by direct computation we get
$$
(\delta_k^3 t_2)^{-\gamma}\delta_k/2\geq (\delta_k^3t_2)^c\ge ((4K_\ast^2)^3t_2)^c =:\mu>1.
$$
Consequently by \eqref{lbk1} we have
$$
b_{k,\gamma}\gtrsim \mu^k.
$$

If $\delta_k\le 4K_\ast^2$, then by direct computation we get
$$
g<(\delta_{k}^6t_2)^{-\gamma}.
$$
Consequently  by \eqref{lbk2} we get
$$
b_{k,\gamma}\gtrsim \|S_k(g)\|
$$
Since $g>f_\ast(\alpha)$, there exists $\tilde \mu>1$ such that for big $k$ we have $\|\mathbf{S}_k(g)\|\ge \tilde \mu^k$.

Thus we conclude that in either case we have
$$
b_{k,\gamma}\gtrsim \mu^k\wedge \tilde \mu^k.
$$
  \hfill $\vartriangleright$

By the claim we have $\gamma\le s_k$ for $k$ big.
Thus   $\gamma\leq s_{\ast}(V).$
Now by the arbitrariness of the choices of $\gamma$ and  $g$ we conclude that
$$
\frac{-\ln f_\ast(\alpha)}{6\ln 4K_\ast^2
+\ln 2(V+5)}=\frac{-\ln f_\ast(\alpha)}{6\ln 4K_\ast^2
+\ln t_2}\le s_*(V)
$$
which is the first inequality in \eqref{K-low-ast}.

Assume $K^\ast<\infty$.
Choose $N\in\N$ such that  $\delta_k<2K^*$ for all $k\geq N.$
By Corollary \ref{K-ast} we have
$0<f^\ast(\alpha)<1$. Fix $g\in (f^\ast(\alpha),1)$.

\medskip

\noindent {\bf  Claim 2.}
 If $\gamma<\frac{-\ln g}{{6\ln 2K_\ast+\ln t_2}},$
 then $\limsup_{k\rightarrow\infty}b_{k,\gamma}=\infty.$

\medskip

\noindent $\vartriangleleft$
For any $k\ge N,$ we have
$$
\gamma<\frac{-\ln g}{6\ln
2K^\ast+\ln t_2}\leq \frac{-\ln g}{6\ln
\delta_k+\ln t_2}.
$$
Consequently $\left(\delta_k^6t_2\right)^{-\gamma}> g.$
By \eqref{2-2} and \eqref{lbk2},
$$b_{k,\gamma}> c^{-2}\left(\frac{(\delta_{k}^6t_2)^{-\gamma}}{g}\right)^{k/2}
\|\mathbf{S}_k(g) \|.$$
This prove the claim by definition of $f^*(\alpha)$ and $g> f^\ast(\alpha)$. \hfill $\vartriangleright$

By the choices of $g$ and $\gamma$ we conclude that
$$
s^*(V)\ge\frac{-\ln f^*(\alpha)}{6\ln2 K^*+\ln t_2}.
$$
which is the first inequality in \eqref{K-up-ast}.

Finally, we consider the cases of $K_*=\infty$ and $K^*=\infty$.
Let us define
$$\hat Q_k(\gamma):= \left(
\begin{array}{ccc}
0&t_2^{-\gamma \underline{a}_k}&0\\
{a}_k/4(t_2a_k/4)^{-\gamma}&0&a_k/4(t_2a_k/4)^{-\gamma}\\
a_k/4(t_2a_k/4)^{-\gamma}&0&{a}_k/4(t_2a_k/4)^{-\gamma}
\end{array}
\right).$$
Fix $\varepsilon_0=1/4$ and take any $w\in\Omega_k(\varepsilon_0)$.
Similar with the proof of Lemma \ref{lem-bc} we can show that
$$
|B_w|\ge \prod_{e_i=e_{12}} \frac{1}{t_2^{a_i-1}}\cdot\prod_{e_i\ne e_{12}}\frac{4}{t_2 a_i}.
$$

Then by a direct computation we get
$$b_{k,\gamma}({\varepsilon_0})\ge
 (1,0,1)\hat Q_1(\gamma) \cdots \hat Q_k(\gamma)(1,1,1)^t.$$
(See \eqref{b-k-beta-epsilon} for the definition of
$b_{k,\gamma}(\varepsilon)$.)

Analogous to the selection procedure in the proof of the claim in Lemma
\ref{infsup}, we can get
$$b_{k,\gamma}\ge b_{k,\gamma}({\varepsilon_0})\gtrsim
\prod_{i=1}^k\frac{a_i}{8}\left(\frac{t_2a_i}{4}\right)^{-\gamma}
=\left(\frac{\delta_k}{8}\left(\frac{t_2\delta_k}{4}\right)^{-\gamma}\right)^k,$$
 By taking $\gamma=s_k$ and using the fact that $b_{k,s_k}=1$ we get
$$s_k\ge\frac{\ln\delta_k-\ln8}{\ln \delta_k+\ln t_2/4}.$$

From this inequality it is seen that  if $K_*=\infty$, then
$s_*(V)\equiv1.$
if $K^*=\infty$, then
$s^*(V)\equiv1.$
 This finish the proof of Proposition \ref{upper-lower}.
\end{proof}


\section{Generating polynomial and Ladders}\label{ladder}

In  this section we give some preparations for the proofs of bounded variation,
bounded covariation and Gibbs like measure.

Consider the equation
$$\Lambda(x,y,z)=x^2+y^2+z^2-xyz-4=V^2.$$
We can solve $z$ as
$$z_{\pm}(x,y,V)=\frac{xy}{2}\pm\frac{1}{2}\sqrt{4V^2+(4-x^2)(4-y^2)}.
$$
For two branches $z=z_+$ or $z=z_-$, let
$$\begin{array}{l}
z_1(x,y,V):=\frac{\partial z(x,y,V)}{\partial x},
\ z_2(x,y,V):=\frac{\partial z(x,y,V)}{\partial y},
\ z_{3}(x,y,V):=\frac{\partial z(x,y,V)}{\partial V}.
\end{array}$$
We also define $z_{ij}(x,y,V)$ by the obvious way.
For any $|x|\leq2$, $|y|\leq2$ and $V>4$, by a simple computation we
get
\begin{equation}\label{dec}
\begin{array}{l}
V-2\le|z_\pm(x,y,V)|\le V+2,\\
|z_j(x,y,V)|\leq1, \ \ \ j=1,2,3\\
|z_{ij}(x,y,V)|\leq1,\ \ \ ij=11,12,13,21,22,23.
\end{array}
\end{equation}

We will estimate the derivatives of generating polynomials by using
Chebishev polynomial $S_p(x)$, which is defined by
$$\begin{array}{l}
S_0(x)\equiv0,\quad S_1(x)\equiv1,\\
S_{p+1}(x)=x S_p(x)-S_{p-1}(x),\quad p\geq1.
\end{array}$$

\subsection{Ladders and modified ladders}\

    In \cite{FLW}, the authors introduce
the notion of ladder and modified ladder
which is very useful for estimating the derivatives of the generating polynomials.
Now we recall the definitions and state some related  results which will be used later.

Given $w\in\Omega_n$, write $w=w_1\cdots w_n$ and $w|_m=w_1\cdots w_m$
for $m=1,\cdots,n.$ Write $B_m=B_{w|_m}.$ Then
for any $k\le n$
$$
B_n\subseteq B_{n-1}\subseteq\cdots\subseteq B_{k}
$$
is  a sequence of spectral generating bands from order $n$ to $k$. We
call the sequence $(B_i)_{i=k}^n$ an {\em initial ladder}, and the
bands $B_i(k\le i\le n)$ are called  {\it initial rungs}. Now we are going
to modify the initial ladder by the following way: for any $i(k\le
i\le n-1)$,
\begin{itemize}
\item if   $B_i$ is of $(i,I)$-type with $a_{i+1}=1$:  delete the rung
$B_{i+1}$ (in this case $B_{i+1}$ must be $(i+1,II)$-type, then
$t_{(i+2,0)}=t_{(i,1)}$ and $t_{(i+1,-1)}=t_{(i,0)}$ implies
$B_{i+1}=B_i$);

\item if $B_i$ is of $(i,I)$-type with $a_{i+1}=2$: change
nothing;

\item if $B_i$ is of $(i,I)$-type with $a_{i+1}>2$: add rungs
$(B_{(i,p)})_{p=2}^{a_{i+1}-1}$ between $B_i$ and $B_{i+1}$ :
$$B_{i+1}=B_{(i,a_{i+1})}\subset B_{(i,a_{i+1}-1)}\subset\cdots\subset
B_{(i,2)}\subset B_{(i,1)}=B_i;$$
where $B_{(i,p)}$ is the unique band in $\sigma_{(i,p)}$ which is included in $B_i$.
\item if $B_i$ is of $(i,II)$ or $(i,III)$-type: change nothing.
\end{itemize}
By this way we get  a unique modified ladder which we relabel as
$$B_n=\hat{B}_m\subset\cdots\subset\hat{B}_1\subset \hat{B}_{0}=B_{k}.$$
We call $(\hat{B}_i)_{i=0}^m$ the {\em modified ladder}, and we
denote the corresponding generating polynomials by
$(\hat{h}_i)_{i=0}^m$. Note that any two consecutive initial rungs
can not be of type $I$ simultaneously,  so we have
\begin{equation}\label{modibd}
(n-k)/2\le m\le  a_{k+1}+a_{k+2}+\cdots+a_{n}.
\end{equation}

Given an initial ladder $(B_i)_{i=k}^n$. Let $(\hat B_i)_{i=0}^m$
be the related modified ladder. For $i=0,\cdots,m-1$ define
\begin{equation}\label{pvalue}
(p_i,l_i)=
\begin{cases}
 (\tau_e(a_{j}),l), &
 \begin{array}{l}
 \text{ if }  (\hat B_{i},\hat B_{i+1})=(B_{j-1},B_{j})\\
 \text{ for some } j \text{ and } w_j=(e,\tau_e(a_j),l)
 \end{array} \\
 (1,1), & \text{  otherwise}.
\end{cases}
\end{equation}
We call $(p_i)_{i=0}^{m-1}$ and $(l_i)_{i=0}^{m-1}$
the {\it type sequence } and {\it index sequence } of the modified ladder.

The following key formula is proved in \cite{FLW}:
\begin{equation}\label{ladder-i}
\hat{h}_{i+1}(x)=z_{\pm}(\hat{h}_i(x),\hat{h}_{i-1}(x),V)
S_{p_i+1}(\hat{h}_{i}(x))-\hat{h}_{i-1}(x)S_{p_i}(\hat{h}_i(x)).
\end{equation}

For convenience  we denote $z_{\pm}(\hat{h}_i(x),\hat{h}_{i-1}(x),V)$ by $z_{\pm}(x).$
By taking derivative on both side of \eqref{ladder-i}, we get
\begin{equation}\label{exprderi}
\begin{array}{r}
\dfrac{\hat{h}'_{i+1}(x)}{\hat{h}'_i(x)}
=S'_{p_i+1}(\hat{h}_i(x))z_\pm(x)-S'_{p_i}(\hat{h}_i(x))\hat{h}_{i-1}(x)+\quad\\
S_{p_i+1}(\hat{h}_i(x))\dfrac{z'_\pm(x)}{\hat{h}'_i(x)}-
S_{p_i}(\hat{h}_i(x))\dfrac{\hat{h}'_{i-1}(x)}{\hat{h}'_i(x)},
\end{array}
\end{equation}
where
\begin{equation}\label{d-cha}
\frac{z'_\pm(x)}{\hat{h}'_i(x)}=
z_1(\hat{h}_i(x),\hat{h}_{i-1}(x),V)+z_2(\hat{h}_i(x),\hat{h}_{i-1}(x),V)
\frac{\hat{h}'_{i-1}(x)}{\hat{h}'_i(x)}.
\end{equation}
We will use this relation later.

Let $p\ge1$, $1\le l\le p$. Define
$$I_{p,l}:=\left\{2\cos\frac{l+c}{p+1}\pi\ :\ |c|\le\frac{1}{10}\mbox{ and }
\ |S_{p+1}(2\cos\frac{l+c}{p+1}\pi)|\le\frac{1}{4}
\right\}.$$

The following property  is  proved in
\cite{FLW} .
\begin{prop}\label{index}
Assume $V\ge20$. Let $(\hat{B}_i)_{i=0}^m$ be a modified ladder,
$(\hat{h}_i)_{i=0}^{m}$ the corresponding generating polynomials,
and $(p_i)_{i=0}^{m-1}, (l_i)_{i=0}^{m-1}$ be the type sequence and index sequence respectively.
Then for any
$0\le i<m$
$$\hat{h}_i(\hat{B}_{i+1})\subset I_{p_i,l_i}.$$
\end{prop}

 We collect some useful
estimations of Chebischev polynomials on the interval $I_{p,l}$,
which is essentially the  Proposition 7 of \cite{LW}.

\begin{prop}[\cite{LW}]\label{keyLW}
Fix $p\ge1$, $1\le l\le p$. For any $t\in I_{p,l}$,
$$\begin{array}{l}
|S_{p+1}(t)|\le\frac{1}{4},\quad |S_p(t)|\le \frac{5}{4},\\[2mm]
\frac{p+1}{3}\csc^2\frac{l\pi}{p+1}
\le |S'_{p+1}(t)|\le (p+1)\csc^2\frac{l\pi}{p+1},\quad
|S'_p(t)|\le2|S'_{p+1}(t)|.\\[2mm]
|S''_{p+1}(t)|\le 4p^2\csc^3\frac{l\pi}{p+1},\quad |S''_{p}(t)|\le 4p^2\csc^3\frac{l\pi}{p+1}.
\end{array}$$
\end{prop}

The following result shown in \cite{FLW} will also be useful later.

\begin{prop}\label{lm-2}
Assume $V\ge20$. Let $(\hat{B}_i)_{i=0}^m$ be a modified ladder,
$(\hat{h}_i)_{i=0}^m$, $(p_i)_{i=0}^{m-1}$ and $(l_i)_{i=0}^{m-1}$
be the corresponding generating polynomials, type sequence and index sequence.
For any $0\le i< m$, $x\in \hat{B}_{i+1}$, we have,
$$\frac{V-8}{3}(p_i+1)\csc^2\frac{l_i\pi}{p_i+1}
\le\left|\frac{\hat{h}'_{i+1}(x)}{\hat{h}'_{i}(x)}\right|
\le (V+5)(p_i+1)\csc^2\frac{l_i\pi}{p_i+1}.$$
\end{prop}

 \smallskip

\noindent {\bf Proof of Lemma \ref{lem-bc}.}\
Given $w\in \Omega_n.$  Consider the initial ladder $(B_i)_{i=0}^n$ with $B_0$
the unique band in $\mathscr{G}_0$  containing $B_w$ and $B_n=B_w.$
 Let $(\hat{B}_i)_{i=0}^m$ be the related  modified ladder and
$(\hat{h}_i)_{i=0}^m$, $(p_i)_{i=0}^{m-1}$ and $(l_i)_{i=0}^{m-1}$
be the corresponding generating polynomials, type sequence and index
sequence.
 Since $\hat h_m(\hat B_m)=[-2,2]$, there exists $x_0\in\hat B_m$
 such that $|\hat h_m^\prime(x_0)| |\hat B_m|=4.$
Notice also that
  $|\hat{h}_{0}^\prime|\equiv 1$(see the explanation after Definition \ref{def1}),
  then by Proposition \ref{lm-2}, the definition of  modified ladder and \eqref{pvalue}
\begin{eqnarray*}
 |B_w|&=&|\hat{B}_m|= 4\frac{|\hat h_0^\prime(x_0)|}{|\hat h_m^\prime(x_0)|}\le 4
\prod_{i=0}^{m-1}\frac{3\sin^2\frac{l_i\pi}{p_i+1}}{(V-8)(p_i+1)}\\
&\le&4\prod_{e_j=e_{12}}\frac{1}{(2t_1)^{a_j-1}}\cdot \prod_{e_j\ne e_{12}}
\frac{1}{(\tau_{e_j}(a_j)+1)t_1}\\
&\le&4\prod_{e_j=e_{12}}\frac{1}{t_1^{a_j-1}}\cdot \prod_{e_j\ne e_{12}} \frac{1}{a_jt_1}.
\end{eqnarray*}
Similarly by using the facts that $\sin x\ge 2x/\pi$ for $x\in[0,\pi/2]$
and $\tau_e(n)+1\le 3n$  we have
\begin{eqnarray*}
 |B_w|&\ge&  4
\prod_{i=0}^{m-1}\frac{\sin^2\frac{l_i\pi}{p_i+1}}{(V+5)(p_i+1)}\\
&\ge&4\prod_{e_j=e_{12}}\frac{1}{(2(V+5))^{a_j-1}}\cdot \prod_{e_j\ne e_{12}}
\frac{4}{(\tau_{e_j}(a_j)+1)^3(V+5)}\\
&\ge&\prod_{e_j=e_{12}}\frac{1}{t_2^{a_j-1}}\cdot \prod_{e_j\ne e_{12}} \frac{1}{a_j^3t_2}.
\end{eqnarray*}

Notice that $a_j\ge 1$ and   $t_1\ge 4$ since $V\ge 20.$
Also notice that any two consecutive  initial rungs can not be type $I$ simultaneously, we get
$$
|B_w|\le 4\prod_{e_j\ne e_{12}} \frac{1}{a_jt_1}\le 4\cdot 4^{-n/2},
$$
which implies the result.
\hfill $\Box$


\section{Bounded variation}\label{bdvar}

The following properties is fundamental for the proof of bounded variation, bounded covariation and continuity of pre-dimensions.

\begin{prop}\label{lip} Assume $V,\widetilde V\ge20$.
Let $(\hat{B}_i)_{i=0}^m$ and $(\hat{\tilde{B}}_i)_{i=0}^m$ be
modified ladders of $\Sigma_{\alpha,V}$ and
$\Sigma_{\alpha,\widetilde V}$ respectively with the same type
sequence $(p_i)_{i=0}^{m-1}$ and index sequence $(l_i)_{i=0}^{m-1}$.
Let $(\hat{h}_i)_{i=0}^m$ and $(\hat{\tilde{h}}_i)_{i=0}^m$ be their
corresponding generating polynomials. Let $x_1\in \hat{B}_{m}$,
$x_2\in \hat{\tilde{B}}_{m}$. Write $q_i=p_i+1$ and
$\theta_i=l_i\pi/q_i.$

(1) If there exist constants $c_1,c_2>1$ and $\lambda>1$ such that
\begin{equation}\label{hcond1}
|\hat{h}_i(x_1)-\hat{\tilde{h}}_i(x_2)|\le
{q_i}^{-1}{\sin^2\theta_i} \left(c_1\lambda^{-i}+c_2|V-\widetilde
V|\right), \quad \forall\ 0<i<m,
\end{equation}
Then there exists  constants $\xi, M\geq 1$  such that
\begin{equation}\label{hresult}
\frac{\left|{\hat{h}'_m(x_1)}/{\hat{h}'_{0}(x_1)}\right|}
{\left|{\hat{\tilde{h}}'_m(x_2)}/\hat{\tilde{h}}'_{0}(x_2)\right|}\le\xi
\exp({M m|V-\widetilde V|}).
\end{equation}
More explicitly we can take
\begin{equation}\label{xi-form}
\begin{array}{rcl}
\xi&=&\exp\left((V+\widetilde
V+130+(4V+53)c_1)(3+\frac{\lambda^2}{(\lambda-1)^2})\right)\\
  M&=&(8V+56)c_2+6.\end{array}
\end{equation}

(2) If there exist constants $c_1, \lambda>1$ such that
\begin{equation}\label{hcond2}
|\hat{h}_i(x_1)-\hat{\tilde{h}}_i(x_2)|\le c_1\
q_i^{-1}\sin^2\theta_i\ \lambda^{-m+i}, \quad \forall\ 0<i<m.
\end{equation}
Then there exists a constant $\xi\ge1$  such that
\begin{equation}\label{hresult1}
\frac{\left|{\hat{h}'_m(x_1)}/{\hat{h}'_{0}(x_1)}\right|}
{\left|{\hat{\tilde{h}}'_m(x_2)}/\hat{\tilde{h}}'_{0}(x_2)\right|}\le\xi
\exp({6 m|V-\widetilde V|}).
\end{equation}
Moreover $\xi$ take the same form as in \eqref{xi-form}.

\end{prop}

\noindent {\bf Proof}. To simplify the notation we write
$$
\Gamma_i:=\left|\frac{\hat{h}'_{i+1}(x_1)}{\hat h'_{i}(x_1)}-
\frac{\hat{\tilde{h}}'_{i+1}(x_2)}{\hat{\tilde{h}}'_{i}(x_2)}\right|\
\ \ \text{ and }\ \ \
  \Delta_i:=|\hat{h}_i(x_1)-\hat{\tilde{h}}_i(x_2)|.
$$

We will prove that  for any $0<i<m$,
\begin{eqnarray}\label{prop2}
\nonumber\Gamma_i&\leq& (4V+23)q_i^2\csc^3\theta_i\cdot \Delta_i+
5q_i\csc^2\theta_i\cdot
\Delta_{i-1}\\
&&+(6q_{i-1}^2)^{-1}\sin^4\theta_{i-1}
\cdot\Gamma_{i-1}+3|V-\widetilde V|q_i\csc^2\theta_i.
\end{eqnarray}

We show first that \eqref{prop2} implies \eqref{hresult} and
\eqref{hresult1}. Notice that by Proposition \ref{lm-2}
\begin{equation}\label{gamma0}
\Gamma_0\leq q_0\csc^2\theta_0(V+\widetilde V+10).
\end{equation}

 \eqref{hresult} and \eqref{hresult1}
  can be shown through the following two claims.

Let $\widetilde M= V+\widetilde V+130+(4V+53)c_1$, $M_1=(4V+28)c_2+3$ and
$M=2M_1.$

 \noindent {\bf Claim 1:} If \eqref{hcond1} holds, then
  for
$0<i<m$
$$
\Gamma_i\leq 2q_i\csc^2\theta_i\left( \widetilde M
(i+1)(6^{-i}+\lambda^{-i})+M_1|V-\widetilde
V|\sum_{l=0}^{i-1}6^{-l}\right).
$$

\noindent $\triangleleft$ Assume first $1<\lambda\leq 6$,
we will show  by induction
$$
\Gamma_i\leq q_i\csc^2\theta_i\left( \widetilde M
(i+1)(6^{-i}+\lambda^{-i})+M_1|V-\widetilde
V|\sum_{l=0}^{i-1}6^{-l}\right).
$$

It is trivially $\Delta_0\leq 4.$ By
using \eqref{prop2} and \eqref{hcond1} for $i=1$,
together with \eqref{gamma0}, we get
\begin{eqnarray*}
\Gamma_1&\leq& (4V+23)q_1^2\csc^3\theta_1\cdot \Delta_1+
5q_1\csc^2\theta_1\cdot
\Delta_{0}\\
&&+(6q_{0}^2)^{-1}\sin^4\theta_{0} \cdot\Gamma_{0}+3|V-\widetilde
V|q_1\csc^2\theta_1\\
&\leq&(4V+23)q_1\csc\theta_1\left(c_1\lambda^{-1}+c_2|V-\widetilde
V|\right)+20q_1\csc^2\theta_1\\
&&+(V+\tilde V+10)/6+3|V-\widetilde V|q_1\csc^2\theta_1\\
&\leq & q_1\csc^2\theta_1\left( 2\widetilde
M(6^{-1}+\lambda^{-1})+M_1|V-\widetilde V|\right).
\end{eqnarray*}

Assume $i>1$ and the statement holds for $i-1$. By \eqref{prop2},
\eqref{hcond1} and $1<\lambda\le6$, we get
\begin{eqnarray*}
\Gamma_i&\leq&
(4V+23)q_i\csc\theta_i\left(c_1\lambda^{-i}+c_2|V-\widetilde
V|\right)\\
&&+5q_i\csc^2\theta_i\left(c_1\lambda^{-i+1}+c_2|V-\widetilde V|\right)\\
&&+ \frac{1}{6}\left( \widetilde M
i(6^{-i+1}+\lambda^{-i+1})+M_1|V-\widetilde
V|\sum_{l=0}^{i-2}6^{-l}\right)
+3|V-\widetilde V|q_i\csc^2\theta_i\\
 &\leq& q_i\csc^2\theta_i\left(
\widetilde M(i+1)(6^{-i}+\lambda^{-i})+M_1|V-\widetilde
V|\sum_{l=0}^{i-1}6^{-l}\right).
\end{eqnarray*}
Thus the result holds for $i$.

Now if $\lambda>6,$ then \eqref{hcond1} fulfills for $\lambda_0=6.$
Thus by what have been proven we get
\begin{eqnarray*}
\Gamma_i&\leq& q_i\csc^2\theta_i\left( \widetilde M
(i+1)(6^{-i}+\lambda_0^{-i})+M_1|V-\widetilde
V|\sum_{l=0}^{i-1}6^{-l}\right)\\
&\leq& 2q_i\csc^2\theta_i\left( \widetilde M
(i+1)(6^{-i}+\lambda^{-i})+M_1|V-\widetilde
V|\sum_{l=0}^{i-1}6^{-l}\right).
\end{eqnarray*}

\hfill $\vartriangleright$

\noindent {\bf Claim 2:} If \eqref{hcond2} holds, then for $0<i<m$
$$
\Gamma_i\leq   q_i\csc^2\theta_i\left(\widetilde
M(6^{-i}+\lambda^{i-m})+3|V-\widetilde
V|\sum_{l=0}^{i-1}6^{-l}\right).
$$

\noindent $\triangleleft$ We show it by induction. By
\eqref{prop2} and \eqref{hcond2} for $i=1$, together with \eqref{gamma0},
\begin{eqnarray*}
\Gamma_1&\leq& (4V+23)q_1^2\csc^3\theta_1\cdot \Delta_1+
5q_1\csc^2\theta_1\cdot
\Delta_{0}\\
&&+(6q_{0}^2)^{-1}\sin^4\theta_{0} \cdot\Gamma_{0}+3|V-\widetilde
V|q_1\csc^2\theta_1\\
 &\le&  q_1\csc^2\theta_1\left(\widetilde
M(6^{-1}+\lambda^{1-m})+3|V-\widetilde V|\right).
\end{eqnarray*}

Assume $i>1$ and the statement holds for $i-1$. By \eqref{prop2} and
induction, if the condition \eqref{hcond2} holds, we can get
\begin{eqnarray*}\label{hinduct-2}
\Gamma_i&\leq& \nonumber
(4V+23)c_1q_i\csc\theta_i\lambda^{i-m}+5c_1q_i\csc^2\theta_i\lambda^{i-1-m}+
\frac{\widetilde M}{6}(6^{1-i}+\lambda^{i-1-m})\\
&&+3|V-\widetilde V|\sum_{l=1}^{i-1}6^{-l}+3|V-\widetilde V|q_i\csc^2\theta_i\\
&\leq& q_i\csc^2\theta_i\left(\widetilde
M(6^{-i}+\lambda^{i-m})+3|V-\widetilde
V|\sum_{l=0}^{i-1}6^{-l}\right),
\end{eqnarray*}
so the conclusion holds. \hfill $\vartriangleright$

The following inequality is basic for us: for any $x,y>0$, we have
\begin{equation}\label{lndiff}
|\ln y-\ln x|\leq (x\wedge y)^{-1}|y-x|.
\end{equation}

As what have been done for  Cookie-cutter set, we have
\begin{eqnarray}\label{thm1-1}
\nonumber&&\left|\ln\left|\frac{\hat{h}'_m(x_1)}{\hat{h}'_0(x_1)}\right|-
\ln\left|\frac{\hat{\tilde{h}}'_m(x_2)}{\hat{\tilde{h}}'_0(x_2)}\right|
\right|\\
\nonumber&=&\left|\sum\limits_{i=0}^{m-1}\left(
\ln\left|\frac{\hat{h}'_{i+1}(x_1)}{\hat{h}'_{i}(x_1)}\right|-
\ln\left|\frac{\hat{\tilde{h}}'_{i+1}(x_2)}
{\hat{\tilde{h}}'_{i}(x_2)}\right|\right)\right|\\
\nonumber&\leq&\sum\limits_{i=0}^{m-1}\left|\left(
\ln\left|\frac{\hat{h}'_{i+1}(x_1)}{\hat{h}'_{i}(x_1)}\right|-
\ln\left|\frac{\hat{\tilde{h}}'_{i+1}(x_2)}
{\hat{\tilde{h}}'_{i}(x_2)}\right|\right)\right|\\
&\leq&\sum\limits_{i=0}^{m-1}\frac{3}{V\wedge\widetilde
V-8}\frac{\sin^2\theta_i}{q_i}\Gamma_i\\
\label{thm1-2}&\leq&\frac{1}{4}\sum\limits_{i=0}^{m-1}\frac{\sin^2\theta_i}{q_i}\Gamma_i\\
\label{thm1-3}&\leq& \widetilde
M\left(1+{36}/{25}+\lambda^2/(\lambda-1)^2\right)+
\begin{cases}
M m|V-\widetilde V|, & \text{by claim } 1\\
6m|V-\widetilde V|, & \text{by claim  } 2,
\end{cases}
\end{eqnarray}
where \eqref{thm1-1} is due to \eqref{lndiff} and Proposition
\ref{lm-2}; \eqref{thm1-2} is due to $V,\widetilde V>20$ and
\eqref{thm1-3} is due to \eqref{gamma0} and the two claims above.
Consequently \eqref{hresult} and \eqref{hresult1} follow.

\medskip

Now fix $0<i<m$, we are going to prove \eqref{prop2}.

For convenience, we denote
$z_{\pm}(\hat{h}_i(x_1),\hat{h}_{i-1}(x_1),V)$ as $z_\pm(x_1)$, and
denote
$z_{\pm}(\hat{\tilde{h}}_i(x_2),\hat{\tilde{h}}_{i-1}(x_2),\widetilde
V)$ as $z_\pm(x_2)$. Both $(\hat{h}_{i+1}(x_1)$, $\hat{h}_i(x_1)$,
$\hat{h}_{i-1}(x_1))$, and $(\hat{\tilde{h}}_{i+1}(x_2)$,
$\hat{\tilde{h}}_i(x_2)$, $\hat{\tilde{h}}_{i-1}(x_2))$ satisfy
\eqref{ladder-i} with the same $p_i$, so by using \eqref{exprderi},
the quantity
$$\begin{array}{l}\frac{\hat{h}'_{i+1}(x_1)}{\hat{h}'_{i}(x_1)}-
\frac{\hat{\tilde{h}}'_{i+1}(x_2)}{\hat{\tilde{h}}'_{i}(x_2)}\end{array}$$
is equal to
\begin{equation}\label{cha0}
\begin{array}{l}
\ \ \ [S'_{p_i+1}(\hat{h}_i(x_1))-S'_{p_i+1}(\hat{\tilde{h}}_i(x_2))]z_\pm(x_1)
 +[z_\pm(x_1)-z_\pm(x_2)]S'_{p_i+1}(\hat{\tilde{h}}_i(x_2))\\
\ \ \ -[S'_{p_i}(\hat{h}_i(x_1))-S'_{p_i}(\hat{\tilde{h}}_i(x_2))]\hat{h}_{i-1}(x_1)
 -[\hat{h}_{i-1}(x_1)-\hat{\tilde{h}}_{i-1}(x_2)]S'_{p_i}(\hat{\tilde{h}}_i(x_2))\\
\ \ \ +[S_{p_i+1}(\hat{h}_i(x_1))-S_{p_i+1}(\hat{\tilde{h}}_i(x_2))]
\frac{z'_\pm(x_1)}{\hat{h}'_i(x_1)}
 +[\frac{z'_\pm(x_1)}{\hat{h}'_i(x_1)}-\frac{z'_\pm(x_2)}{\hat{\tilde{h}}'_i(x_2)}]
 S_{p_i+1}(\hat{\tilde{h}}_i(x_2))\\
\ \ \ -[S_{p_i}(\hat{h}_i(x_1))-S_{p_i}(\hat{\tilde{h}}_i(x_2))]
\frac{\hat{h}'_{i-1}(x_1)}{\hat{h}'_i(x_1)}
-[\frac{\hat{h}'_{i-1}(x_1)}{\hat{h}'_i(x_1)}-
\frac{\hat{\tilde{h}}'_{i-1}(x_2)}{\hat{\tilde{h}}'_i(x_2)}]
S_{p_i}(\hat{\tilde{h}}_i(x_2)).
\end{array}
\end{equation}
There are eight terms in (\ref{cha0}), we will estimate them one by
one.

By proposition \ref{index},  $\hat{h}_i(x_1),\hat{\tilde{h}}_i(x_2)\in I_{p_i,l_i}$,  thus by
Proposition \ref{keyLW},
\begin{equation}\label{cha1}
\begin{array}{l}
\left|S_{p_i+1}(\hat{h}_i(x_1))-S_{p_i+1}(\hat{\tilde{h}}_i(x_2))\right|
\le q_i\csc^2\theta_i \Delta_i\\[8pt]
\left|S'_{p_i+1}(\hat{h}_i(x_1))-S'_{p_i+1}(\hat{\tilde{h}}_i(x_2))\right|
\le 4 q_i^2\csc^3\theta_i \Delta_i\\[8pt]
\left|S_{p_i}(\hat{h}_i(x_1))-S_{p_i}(\hat{\tilde{h}}_i(x_2))\right|
\le 2q_i\csc^2\theta_i \Delta_i\\[8pt]
\left|S'_{p_i}(\hat{h}_i(x_1))-S'_{p_i}(\hat{\tilde{h}}_i(x_2))\right|
\le 4 q_i^2\csc^3\theta_i \Delta_i.
\end{array}
\end{equation}

By \eqref{dec} and \eqref{d-cha}, we have
\begin{equation}\label{cha1-1}
\begin{array}{rcl}
\left|z_\pm(x_1)-z_\pm(x_2)\right|
&\le& \Delta_i+\Delta_{i-1}+|V-\widetilde V|\\
\left|\frac{z'_{\pm}(x_1)}{\hat{h}'_i(x_1)}\right|&=&
\left|z_1(x_1)+z_2(x_1)\frac{\hat{h}_{i-1}'(x_1)}{\hat{h}_i'(x_1)}\right|\ \le\ 2\\[6pt]
\left|\frac{z'_{\pm}(x_1)}{\hat{h}'_i(x_1)}-
\frac{z'_{\pm}(x_2)}{\hat{\tilde{h}}'_i(x_2)}\right|&\le&
|z_1(x_1)-z_1(x_2)|+\left|(z_2(x_1)-z_2(x_2))\frac{\hat{h}_{i-1}'(x_1)}{\hat{h}_i'(x_1)}\right|\\
&&+\left|z_2(x_2)(\frac{\hat{h}_{i-1}'(x_1)}{\hat{h}_i'(x_1)}-
\frac{\hat{\tilde{h}}_{i-1}'(x_2)}{\hat{\tilde{h}}_i'(x_2)})\right|
\end{array}
\end{equation}
and
\begin{equation}\label{cha2}
\begin{array}{l}
|z_1(\hat{h}_i(x_1),\hat{h}_{i-1}(x_1),V)-
z_1(\hat{\tilde{h}}_i(x_2),\hat{\tilde{h}}_{i-1}(x_2),\widetilde
V)|\le
\Delta_i+\Delta_{i-1}+|V-\widetilde V|\\[5pt]
|z_2(\hat{h}_i(x_1),\hat{h}_{i-1}(x_1),V)-
z_2(\hat{\tilde{h}}_i(x_2),\hat{\tilde{h}}_{i-1}(x_2),\widetilde
V)|\le
\Delta_i+\Delta_{i-1}+|V-\widetilde V|.\\
\end{array}\end{equation}
By a direct computation and Proposition \ref{lm-2},
\begin{equation}\label{cha3}
\begin{array}{rcl}
\left|\dfrac{\hat{h}'_{i-1}(x_1)}{\hat{h}'_i(x_1)}-
\dfrac{\hat{\tilde{h}}'_{i-1}(x_2)}{\hat{\tilde{h}}'_i(x_2)}\right|
&=&\left|\dfrac{\hat{h}'_{i-1}(x_1)}{\hat{h}'_i(x_1)}
\dfrac{\hat{\tilde{h}}'_{i-1}(x_2)}{\hat{\tilde{h}}'_i(x_2)}\right|
\Gamma_{i-1}\\
&\le&\dfrac{9\sin^4\theta_{i-1}}{(V-8)(\widetilde
V-8)q_{i-1}^2}\Gamma_{i-1}.
\end{array}
\end{equation}

Now we estimate the eight terms in \eqref{cha0} one by one. By
\eqref{dec} and \eqref{cha1}, the first term is bounded by
$$4 (V+2)q_i^2\csc^3\theta_i \Delta_i.$$
By \eqref{cha1-1} and Proposition \ref{keyLW}, the second term is
bounded by
$$q_i\csc^2\theta_i(\Delta_i+\Delta_{i-1}+|V-\widetilde V|).$$
By \eqref{cha1} and $|\hat{h}_{i-1}(x_1)|\le2$, the third term is bounded by
$$8q_i^2\csc^3\theta_i\Delta_i.$$
By Proposition \ref{keyLW}, the 4th term is bounded by
$$2q_i\csc^2\theta_i\Delta_{i-1}.$$
By \eqref{cha1} and \eqref{cha1-1}, the 5th term is bounded by
$$2q_i\csc^2\theta_i\Delta_i.$$
By \eqref{dec} and Proposition \ref{keyLW}, $|z_2(x_2)|\le1$ and
$|S_{p_i+1}(\hat{\tilde{h}}_{i}(x_2))|\le 1/4$, then by
$|\frac{\hat{h}'_{i-1}(x_1)}{\hat{h}'_i(x_1)}|\le 1/4$,
\eqref{cha1-1}, \eqref{cha2} and \eqref{cha3}, the 6th term is
bounded by
$$\begin{array}{l}
2\left(\Delta_i+ \Delta_{i-1}+|V-\widetilde V|\right)+
\dfrac{9\sin^4\theta_{i-1}}{(V-8)(\widetilde
V-8)q_{i-1}^2}\Gamma_{i-1}.
\end{array}$$
By \eqref{cha1} and $|\hat{h}'_{i-1}(x_1)|/|\hat{h}'_i(x_1)|\le
1/4$, the 7th term is bounded by
$$2q_i\csc^2\theta_i\Delta_i.$$
By Proposition \ref{keyLW} and \eqref{cha3}, the 8th term is bounded by
$$\dfrac{45\sin^4\theta_{i-1}}{4(V-8)(\widetilde V-8)q_{i-1}^2}\Gamma_{i-1}.$$
Take sum on the eight bounds, we get \eqref{prop2}. This proves the
proposition. \hfill $\Box$ \vskip 0.2cm

\noindent {\bf Proof of Theorem \ref{bvar}}.
Let $$B_n\subset B_{n-1}\subset\cdots\subset B_0$$ be a sequence of
spectral generating bands (with orders from $n$ to $0$), which form
an initial ladder. Let $(\hat{B}_i)_{i=0}^m$ be the corresponding
modified ladder, $(\hat{h}_i)_{i=0}^m$ the corresponding generating
polynomials. Note that $\hat{B}_0=B_0$ and
$\hat{h}'_0\equiv1$. To apply Proposition \ref{lip}, we only need to
verify \eqref{hcond2}.

Let $\lambda:=(V-8)/3.$ For any $0\leq i< m$ and $x,y\in
\hat{B}_{i+1}$, since $\hat{h}_i$ is monotone on $\hat{B}_i$, we
have
$$\begin{array}{rcl}
|\hat{h}_i(x)-\hat{h}_i(y)|&
=&\left|\displaystyle\int_x^y \frac{\hat{h}'_i(t)}
{\hat{h}'_{i+1}(t)} \hat{h}'_{i+1}(t) dt\right|\\
&\le& \lambda^{-1} q_i^{-1}\sin^2\theta_i \left|
\displaystyle\int_x^y \hat{h}'_{i+1}(t) dt\right|\\
&=& \lambda^{-1} q_i^{-1}\sin^2\theta_i |\hat{h}_{i+1}(x)-\hat{h}_{i+1}(y)|,\\
\end{array}$$
where the inequality is due to Proposition \ref{lm-2}.
Since $\hat{h}_m(\hat{B}_m)=[-2,2]$, for any $x_1,x_2\in \hat{B}_m$,
$|\hat{h}_m(x_1)-\hat{h}_m(x_2)|\le4$, hence we have
$$|\hat{h}_i(x_1)-\hat{h}_i(x_2)|\le 4\lambda^{i-m}
\prod_{l=i}^{m-1}q_l^{-1}\sin^2\theta_l\leq
4q_i^{-1}\sin^2\theta_i\lambda^{i-m},\quad \forall\, 0\leq i<m.$$

Now by Proposition \ref{lip} the result follows for some constant
$\xi$ which only depends on $V$. More explicitly notice that
$\lambda>4$, then by \eqref{xi-form} we can take
$$
\xi=\exp\left( 180V\right).
$$

\hfill $\Box$ \vskip 0.2cm

\noindent {\bf Proof of Corollary \ref{bdist}}\ Write $B=[a,b]$.
We know that $h$ is monotone on $B$ and $h(B)=[-2,2].$
By mean value theorem, there exists $x_0\in B$ such that
$$
4=|h(a)-h(b)|=|h^\prime(x_0)||B|.
$$
By Proposition \ref{bvar}, the result follows.
\hfill $\Box$


\section{Bounded covariation }\label{bdcov}

\begin{prop}\label{diff}
Assume $V,\widetilde V\ge 24.$ Suppose that $(\hat{B}_i)_{i=0}^m$ and
$(\hat{\tilde{B}}_i)_{i=0}^m$ are  modified ladders of
$\Sigma_{\alpha,V}$ and $\Sigma_{\alpha,\widetilde V}$
respectively with the same type sequence $(p_i)_{i=0}^{m-1}$, and
the same index sequence $(l_i)_{i=0}^{m-1}$. Let
$(\hat{h}_i)_{i=0}^m$ and $(\hat{\tilde{h}}_i)_{i=0}^m$ be their
corresponding generating polynomials. Then  we have
\begin{itemize}
\item[{\rm(i)}]\ There exists positive
 constant $c$ depending only on $V$ with $c\le18/37$ such that
for any $0<i<m$ and any $x_1\in \hat{B}_{i+1}$, $x_2\in \hat{\tilde{B}}_{i+1}$,
\begin{equation}\label{ediff}
 \Delta_i\le q_i^{-1}{\sin^2\theta_i}
\left(c\Delta_{i+1}+ c\Delta_{i-1}+|V-\widetilde V|\right).
\end{equation}
where $\Delta_i=|\hat{h}_{i}(x_1)-\hat{\tilde{h}}_i(x_2)|,
q_i=p_i+1$ and $\theta_i=l_i\pi/q_i.$

\item[{\rm (ii)}]\ Let $\lambda=\frac{1+\sqrt{1-4c^2}}{2c}$. Then there
exist absolute constants $c_1,c_2>1$ such that for any $x_1\in
\hat{B}_{m}$, there exists $x_2\in \hat{\tilde{B}}_{m}$ such that,
\begin{equation}\label{distdist}
\Delta_i\le  q_i^{-1}{\sin^2\theta_i}
 \left(c_1\lambda^{-i}+c_2|V-\widetilde V|\right), \quad 0< i< m.
\end{equation}
\item[{\rm (iii)}]\ There exists absolute constants $C_1,C_2,C_3>1$ such that
$$\eta^{-1} \frac{|\hat{\tilde{B}}_m|}{|\hat{\tilde{B}}_{0}|}\le
\frac{|\hat{B}_m|}{|\hat{B}_{0}|}\le \eta
\frac{|\hat{\tilde{B}}_m|}{|\hat{\tilde{B}}_{0}|},$$ where
$\eta=C_1\exp\left(C_2(V+\widetilde V)+C_3m|V-\widetilde V|\right).$
\end{itemize}
\end{prop}

\noindent {\bf Proof}. (i) Take $0<i<m$ and $x_1\in \hat{B}_{i+1}$,
$x_2\in \hat{\tilde{B}}_{i+1}$.

For convenience, we denote
$z_{\pm}(\hat{h}_i(x_1),\hat{h}_{i-1}(x_1),V)$ as $z_\pm(x_1)$, and
also denote
$z_{\pm}(\hat{\tilde{h}}_i(x_2),\hat{\tilde{h}}_{i-1}(x_2),V)$ as
$z_\pm(x_2)$. Both $(\hat{h}_{i+1}(x_1)$, $\hat{h}_i(x_1)$,
$\hat{h}_{i-1}(x_1))$, and $(\hat{\tilde{h}}_{i+1}(x_2)$,
$\hat{\tilde{h}}_i(x_2)$, $\hat{\tilde{h}}_{i-1}(x_2))$ satisfy
\eqref{ladder-i} with the same $p_i$. So, we have
$$
\begin{array}{rcl}
\hat{h}_{i+1}(x_1)-\hat{\tilde{h}}_{i+1}(x_2)
&=&z_{\pm}(x_1)[S_{p_i+1}(\hat{h}_{i}(x_1))-S_{p_i+1}(\hat{\tilde{h}}_{i}(x_2))]\\
&&+[z_{\pm}(x_1)-z_{\pm}(x_2)]S_{p_i+1}(\hat{\tilde{h}}_{i}(x_2))\\
&&-\hat{h}_{i-1}(x_1)[S_{p_i}(\hat{h}_{i}(x_1))-S_{p_i}(\hat{\tilde{h}}_{i}(x_2))]\\
&&-[\hat{h}_{i-1}(x_1)-\hat{\tilde{h}}_{i-1}(x_2)]S_{p_i}(\hat{\tilde{h}}_{i}(x_2)).
\end{array}
$$

By Proposition \ref{index}, $\hat{h}_i(x_1)$,
$\hat{\tilde{h}}_i(x_2)\in I_{p_i,l_i}$, then by Proposition
\ref{keyLW},
$$\left|S_{p_i+1}(\hat{h}_i(x_1))-S_{p_i+1}(\hat{\tilde{h}}_i(x_2))\right|
\ge \frac{q_i}{3}\cdot \csc^2\theta_i \cdot \Delta_i.
$$
By Proposition \ref{keyLW} again,
$$
\left|S_{p_i}(\hat{h}_i(x_1))-S_{p_i}(\hat{\tilde{h}}_i(x_2))\right|
\le 2q_i \csc^2\theta_i \Delta_i.
$$
So by Proposition \ref{keyLW},
\eqref{cha1-1} and the above three formulas, we have
$$
\Delta_{i+1} \ge (\frac{V-2}{3}-4) q_i\csc^2\theta_i \Delta_i
-\frac{1}{4}\Delta_i
 -\frac{3}{2}\Delta_{i-1}-\frac{1}{4}|V-\widetilde V|.
$$
 Let $c:=18/(4V-59)$, then $0<c\le 18/37$ since $V\ge24$. We get
$$
\Delta_i\leq
q_i^{-1}\sin^2\theta_i\left(c\Delta_{i+1}+c\Delta_{i-1}+
|V-\widetilde V|\right).
$$

\medskip

(ii) Take any $x_1\in \hat{B}_m$. By
$\hat{h}_m(\hat{B}_m)=[-2,2]$, $\hat{\tilde{h}}_m(\hat{\tilde{B}}_m)=[-2,2]$,
there exists $x_2\in\hat{\tilde{B}}_m$ such that
$$\hat{h}_m(x_1)=\hat{\tilde{h}}_m(x_2).$$
Thus $\triangle_m=0$. Take any integer $i\in\{0,\cdots,m-1\}$. By
$\hat{h}_i(\hat{B}_m)\subset[-2,2]$,
$\hat{\tilde{h}}_i(\hat{\tilde{B}}_m)\subset[-2,2]$, we get
$$\Delta_i=|\hat{h}_i(x_1)-\hat{\tilde{h}}_i(x_2)|\le4.$$

 \eqref{ediff} implies that for $0<i<m$
\begin{equation}\label{Delta-i}
\triangle_{i}\le c(\triangle_{i+1}+\triangle_{i-1})+|V-\widetilde
V|.
\end{equation}
Notice that $\lambda=\frac{1+\sqrt{1-4c^2}}{2c}\ge 37/36$ is the larger
root of $x^2-x/c+1=0$. Write $c^\prime=\lambda+\lambda^{-1},$ hence
\eqref{Delta-i} can be rewritten as
\begin{equation}\label{ite-delta}
\lambda\triangle_i-\triangle_{i+1}\le
 \lambda^{-1}(\lambda \triangle_{i-1}-\triangle_{i})+c^\prime|V-\widetilde V|.
\end{equation}

\noindent{\bf Claim:} For $0< i< m$ we have
$$\Delta_i\leq
8\lambda^{-i}\sum_{k=0}^\infty\lambda^{-2k}+ c^\prime |V-\widetilde
V|\sum_{k=1}^\infty k\lambda^{-k}.
$$

\noindent $\vartriangleleft$ We show it by induction.

At first by using \eqref{ite-delta}, for $0<i<m$ we can get
$$
\lambda\Delta_i-\Delta_{i+1}\leq
\lambda^{-i}(\lambda\Delta_0-\Delta_1)+c^\prime|V-\widetilde
V|\sum_{k=0}^{i-1}\lambda^{-k}.
$$
Take $i=m-1$ and notice that $\Delta_i\leq 4, \Delta_m=0$ we get
$$
\Delta_{m-1}\leq 8\lambda^{1-m}+ c^\prime|V-\widetilde
V|\sum_{j=1}^{m-1}\lambda^{-j}.
$$

Assume  the result holds for $i+1$. Then
\begin{eqnarray*}
\Delta_i&\leq&
\lambda^{-1}\left(\Delta_{i+1}+\lambda^{-i}(\lambda\Delta_0-\Delta_1)+
c^\prime|V-\widetilde V|\sum_{k=0}^{i-1}\lambda^{-k}\right)\\
&\le&8\lambda^{-i}\sum_{k=0}^\infty\lambda^{-2k}+ c^\prime
|V-\widetilde V|\sum_{k=1}^\infty k\lambda^{-k}.
\end{eqnarray*}
\hfill $\vartriangleright$

Thus for $0< i<m$ we have
$$\Delta_i\leq
M_1\lambda^{-i} + c^\prime  M_2|V-\widetilde V|
$$
with $M_1=8\lambda^2/(\lambda^2-1)$ and $ M_2=\sum_{k=1}^\infty
k\lambda^{-k}=\lambda/(\lambda-1)^2.$ Since $\Delta_0\leq 4$ and
$\Delta_m=0,$ the inequality also holds for $i=0,m.$

Consequently for $0<i<m,$ by \eqref{ediff} we get
\begin{eqnarray*}
\Delta_i&\leq&
q_i^{-1}\sin^2\theta_i\left(c\Delta_{i+1}+c\Delta_{i-1}+
|V-\widetilde V|\right)\\
&\leq&q_i^{-1}\sin^2\theta_i
\left(M_1c(\lambda^{-(i+1)+\lambda^{-(i-1)}})+(1+2cc^\prime
M_2)|V-\widetilde V|
\right)\\
&\leq&q_i^{-1}\sin^2\theta_i\left(M_1\lambda^{-i}+(1+c^\prime
M_2)|V-\widetilde V| \right).
\end{eqnarray*}
Recall that  $\lambda\ge 37/36=: \lambda_0>1$, thus
$$
M_1\leq \frac{8\lambda_0^2}{\lambda_0^2-1}=:c_1\ \ \text{ and }\ \
1+c^\prime M_2\leq 1+\frac{\lambda_0^2+1}{(\lambda_0-1)^2}=: c_2.
$$
Consequently \eqref{distdist} holds with two absolute constants $c_1$ and $c_2$.

\bigskip

(iii) Proposition \ref{lip} and \eqref{distdist} imply that there
exist absolute constants $C_1^\prime,C_2^\prime,C_3^\prime>1$  such
that, for any $\hat{x}\in\hat{B}_m$, there exists
$\hat{y}\in\hat{\tilde{B}}_m$ such that
\begin{equation}\label{ratio}
\xi_1^{-1}\le\left|\frac{\hat{h}'_m(\hat{x})/\hat{h}'_0(\hat{x})}
{\hat{\tilde{h}}'_m(\hat{y})/\hat{\tilde{h}}'_0(\hat{y})}\right|\le
\xi_1,
\end{equation}
where $\xi_1=C_1^\prime\exp\left(C_2^\prime(V+\widetilde
V)+C_3^\prime Vm|V-\widetilde V|\right)$.

By the definition of generating polynomial, there exist
$x_1\in \hat{B}_m$, $x_2\in\hat{B}_0$ such that
$$|\hat{B}_m|\,|\hat{h}_m'(x_1)|=4,\ |\hat{B}_0|\,|\hat{h}_0'(x_2)|=4.$$

By Theorem \ref{bvar} and \ref{bdist}, we have
$$\frac{|\hat{B}_m|}{|\hat{B}_0|}=
\frac{|\hat{B}_m|\,|\hat{h}_m'(x_1)|}{|\hat{B}_0|\,|\hat{h}_0'(x_2)|}
\frac{|\hat{h}'_m(\hat{x})|}{|\hat{h}'_m(x_1)|}\,
\frac{|\hat{h}'_0(x_2)|}{|\hat{h}'_0(\hat{x})|}\,
\frac{|\hat{h}'_0(\hat{x})|}{|\hat{h}'_m(\hat{x})|} \le
16\exp(720V)\left|\frac{\hat{h}'_0(\hat{x})}{\hat{h}'_m(\hat{x})}\right|.
$$

By the same discussion, we have
$$\frac{|\hat{\tilde{B}}_m|}{|\hat{\tilde{B}}_0|}
\ge \frac{1}{16}\exp(-720\widetilde V)\left|
\frac{\hat{\tilde{h}}'_0(\hat{y})}{\hat{\tilde{h}}'_m(\hat{y})}\right|.
$$

Then by \eqref{ratio}, we have
$$\frac{|\hat{B}_m|}{|\hat{B}_0|}\le \eta\frac{|\hat{\tilde{B}}_m|}{|\hat{\tilde{B}}_0|}$$
with $\eta=C_1\exp\left(C_2(V+\widetilde
V)+C_3 Vm|V-\widetilde V|\right)$, where $C_1,C_2,C_3$ are still absolute constants.

The opposite direction of the inequality can be got by the same way.
\hfill $\Box$ \vskip 0.2cm

 \noindent {\bf Proof of Theorem \ref{bco}}\
 This is a direct consequence of Proposition \ref{diff} (iii).
 \hfill $\Box$

\noindent {\bf Proof of Corollary \ref{C-n}.}\
For each $n\in\mathscr{N}$, fix some $w^{(n)}\in \Omega_{l_n}$
such that $B_{w^{(n)}}$ is of type I and $a_{l_n+1}=n.$ Then define
\begin{equation}\label{xi-n}
\zeta_n:=\frac{|B_{w^{(n)}u}|}{|B_{w^{(n)}}|}.
\end{equation}
By applying Theorem \ref{bco}, we get the result.
If moreover $a_{l_{n}+1}=1$, then we know that $B_{w^{(n)}u}=B_{w^{(n)}}$, thus $\zeta_1$=1.
\hfill $\Box$

Now we can give the proof of Theorem \ref{lip-conti}.

\begin{prop}\label{cont-1}
For $V,\widetilde V\ge 24$, there exists an absolute constant $C>0$
such that
$$\begin{array}{l}
|s_*(V)-s_*(\widetilde V)|\le CV|V-\tilde V|,\\
|s^*(V)-s^*(\widetilde V)|\le CV|V-\tilde V|.\end{array}$$
\end{prop}

\begin{proof}
For any $w\in\Omega_{n}$, let $B_w$ and $\tilde
B_w$ be the related bands of $\Sigma_{\alpha,V}$ and
$\Sigma_{\alpha,\widetilde V}$ respectively. Let $(B_i)_{i=0}^n$
and $(\tilde B_i)_{i=0}^n$ be the ladders  of $\Sigma_{\alpha,V}$
and $\Sigma_{\alpha,\widetilde V}$ respectively with
$B_n=B_w$ and $\tilde B_n=\tilde B_w.$ Let $(\hat
B_i)_{i=0}^{m_w}$ and $(\hat{\tilde B}_i)_{i=0}^{m_w}$ be
the related modified ladders. Then  by \eqref{modibd} we have $m_w\geq n/2$.  By \eqref{upper-bd},
$$
|B_w|,|\tilde B_w|\le 4^{1-m_w}.
$$

By Proposition \ref{diff} (iii), there exists absolute constants
$C_1,C_2,C_3>1$ such that
$$
\eta^{-1} \le
\frac{|{B}_w|}{|\tilde {B}_{w}|}\le \eta ,
$$ where
$\eta=C_1\exp\left(C_2(V+\widetilde V)+C_3Vm_\omega|V-\widetilde
V|\right).$

Write $s_n=s_n(V)$ and $\tilde s_n=s_n(\widetilde V).$ Let
$d:=\limsup_{n\to\infty}|s_n-\tilde s_n|$, then $d\leq 1$ and it is easy to show that
$$
|s_\ast(V)- s_\ast(\widetilde V)|, |s^\ast(V)- s^\ast(\widetilde
V)|\leq d.
$$

If $d=0$ the result holds trivially. So in the following we assume
$d>0$. Then there exist infinitely many $n$ such that $s_n\geq
\tilde s_n+d/2$ or $\tilde s_n\ge s_n+d/2.$

 At first we assume that there are
infinitely many $n$ such that $s_n\ge\tilde s_n +d/2$.  For those
$n$ big enough we have
\begin{eqnarray*}
1&=&\sum_{|w|=n}|B_w|^{s_n}\leq
\sum_{|w|=n}|B_w|^{\tilde s_n+d/2}\\
&\le& \sum_{|w|=n}\eta^{\tilde s_n+d/2}|\tilde
B_w|^{\tilde s_n+d/2} \le \sum_{|w|=n}\eta^{\tilde
s_n+d/2}4^{(1-m_w)d/2}
|\tilde B_w|^{\tilde s_n}\\
&\le& C(V,\widetilde V)\sum_{|w|=n} \exp[-m_w \left(d\ln2-C_3V(\tilde
s_n+d/2)|V-\tilde V|\right)]|\tilde B_w|^{\tilde s_n}.
\end{eqnarray*}

We claim that $d\ln 2\le 2C_3V|V-\widetilde V|.$ In fact if otherwise,
notice that $\tilde s_n, d\leq 1$ and $m_\omega\ge n/2$, we should
get
$$
1\le C(V,\widetilde V) \exp[-C_3n|V-\tilde V|/4]\sum_{|\omega|=n}|\tilde
B_\omega|^{\tilde s_n}=C(V,\widetilde V) \exp[-C_3n|V-\tilde V|/4],
$$ which leads to contradiction for large
$n$. So we have
$$
d \leq \frac{2C_3V|V-\tilde V|}{\ln2}.
$$

For the case that there are infinitely many $n$ such that $\tilde
s_n\geq s_n+d/2$, the argument is the same.
\end{proof}

\noindent {\bf Proof of Theorem \ref{lip-conti}}
It is a direct consequence of Proposition \ref{cont-1}.
\hfill $\Box$


\section{Gibbs like measure}\label{gibbs-meas}
Throughout this section we take $V\ge24$, $0\le\varepsilon<1/12$ and consider the set $E_\varepsilon$ defined in \eqref{truncE}.
We will construct a
 Gibbs like measure on $E_\varepsilon.$

For any $m\ge k$, $T=I,\ II,$ or $III$, define
$$
\Omega_m^{(k,T)}(\varepsilon)=\{w\in \Omega_m(\varepsilon):  e_{w_k}=(\ast,T)\},
$$
where $\Omega_m(\varepsilon)$ is defined in \eqref{truncO}.
For any $0<\beta<1$ define
$$
b_{m,\beta}^{(k,T)}(\varepsilon)=\sum_{w\in\Omega_{m}^{(k,T)}(\varepsilon)}|B_w|^\beta .
$$

Fix $0<\beta<1.$ At first we discuss  the relationship between
$b_{k-1,\beta}(\varepsilon)$ and $b_{k,\beta}(\varepsilon)$
(see \eqref{b-k-beta-epsilon} for definition).
As a preparation we define the following
sequence
$$
A_{0,\beta}(\varepsilon):=0\ \ \ A_{n,\beta}(\varepsilon):=
\sum_{(n+1)\varepsilon< j<(n+1)(1-\varepsilon)}
\frac{1}{(n+1)^\beta}\sin^{2\beta}\frac{j\pi}{n+1}\ \ \  (n\geq 1).
$$

\begin{lem}\label{A-n-beta}
$A_{n,\beta}:=A_{n,\beta}(0)\sim
 A_{n,\beta}(\varepsilon)\sim \ n^{1-\beta}.
$ And $A_{n,\beta}\sim A_{n+1,\beta}$ for $n\geq 1.$
\end{lem}

\begin{proof}
Since
\begin{eqnarray*}
A_{n,\beta}(\varepsilon)&=&
(n+1)^{1-\beta}\sum_{(n+1)\varepsilon< j
< (n+1)(1-\varepsilon)}\frac{1}{n+1}\sin^{2\beta}\frac{j\pi}{n+1}\\
&\sim&\frac{(n+1)^{1-\beta}}{\pi}
\int_{\varepsilon\pi}^{(1-\varepsilon)\pi}\sin^{2\beta} xdx.
\end{eqnarray*}
Since $\varepsilon<1/12$, the result follows.
\end{proof}

\begin{rem}\label{const}{\rm
Here the constants related to $``\sim"$ only depend on $\beta.$}
\end{rem}

\begin{prop}
 For any
$k\ge1$, we  have
\begin{equation}\label{three}
\frac{b_{k,\beta}^{(k,I)}(\varepsilon)}{b_{k,\beta}(\varepsilon)}\sim 1;\ \ \
\frac{b_{k,\beta}^{(k,II)}(\varepsilon)}{b_{k,\beta}(\varepsilon)}\sim
\frac{\zeta_{a_k}^{\beta}}{A_{a_k,\beta}};\ \ \
\frac{b_{k,\beta}^{(k,III)}(\varepsilon)}{b_{k,\beta}(\varepsilon)}\sim
\begin{cases}
1 & a_k>1\\
\frac{\zeta_{a_{k-1}}^{\beta}}{A_{a_{k-1},\beta}}& a_k=1
\end{cases},
\end{equation}
(where $\xi_n$ is defined in \eqref{xi-n}) and
\begin{equation}\label{b-k-b-k-1}
\frac{b_{k,\beta}(\varepsilon)}{b_{k-1,\beta}(\varepsilon)}\sim A_{a_k,\beta}.
\end{equation}
Moreover the constants related to $``\sim"$ only depend on $V$ and
$\beta.$
\end{prop}

\begin{proof}
By the definition and Lemma \ref{A-n-beta}, $A_{n,\beta}\sim
A_{n+1,\beta}$ for $n\geq 1$. By Proposition \ref{lm-2} we have
\begin{equation}\label{b-k-1}
\begin{array}{rcl}
b_{k,\beta}^{(k,I)}(\varepsilon)&=&\sum_{w\in\Omega_{k}^{(k,I)}(\varepsilon)}
|B_w|^\beta\\
&=&\sum_{w\in\Omega_{k-1}^{(k-1,II)}(\varepsilon)}
\sum_{j=\lceil\varepsilon(a_k+2)\rceil}^{\lfloor (1-\varepsilon)(a_k+2)\rfloor}
|B_{w\ast(e_{21},a_{k}+1,j)}|^\beta  \\
&&+ \sum_{w\in\Omega_{k-1}^{(k-1,III)}(\varepsilon)}
\sum_{j=\lceil\varepsilon (a_k+1)\rceil}^{\lfloor (1-\varepsilon)(a_k+1)\rfloor}
|B_{w\ast( e_{31},a_k,j)}|^\beta  \\
&\sim&\sum_{w\in\Omega_{k-1}^{(k-1,II)}(\varepsilon)}|B_w|^\beta
\sum_{j=\lceil\varepsilon(a_k+2)\rceil}^{\lfloor (1-\varepsilon)(a_k+2)\rfloor}
(a_k+2)^{-\beta}\sin^{2\beta}\frac{j\pi}{a_k+2}\\
&&+
\sum_{w\in\Omega_{k-1}^{(k-1,III)}(\varepsilon)}|B_w|^\beta
\sum_{j=\lceil\varepsilon (a_k+1)\rceil}^{\lfloor (1-\varepsilon)(a_k+1)\rfloor}
(a_k+1)^{-\beta}\sin^{2\beta}\frac{j\pi}{a_k+1}\\
&=&A_{a_k+1,\beta}\cdot b_{k-1,\beta}^{(k-1,II)}(\varepsilon) + A_{a_k,\beta}
\cdot b_{k-1,\beta}^{(k-1,III)}(\varepsilon)\\
&\sim&A_{a_k,\beta}\left( b_{k-1,\beta}^{(k-1,II)}(\varepsilon) +
b_{k-1,\beta}^{(k-1,III)}(\varepsilon)\right).
\end{array}
\end{equation}

Similarly we have
\begin{equation}\label{b-k-3}
b_{k,\beta}^{(k,III)}(\varepsilon)
\ \ \sim \ \ A_{a_k,\beta}\cdot b_{k-1,\beta}^{(k-1,II)}(\varepsilon) +
A_{a_k-1,\beta} \cdot b_{k-1,\beta}^{(k-1,III)}(\varepsilon).
\end{equation}
We can see
\begin{equation}\label{b-k-13}
\begin{array}{ll}
b_{k,\beta}^{(k,I)}(\varepsilon)\sim b_{k,\beta}^{(k,III)}(\varepsilon),& \mbox{ if } a_k>1,\\
b_{k,\beta}^{(k,I)}(\varepsilon)\gtrsim b_{k,\beta}^{(k,III)}(\varepsilon),&
\mbox{ if } a_k=1.\\
\end{array}
\end{equation}

On the other hand by Corollary \ref{C-n} we get

\begin{equation}\label{b-k-2}
\begin{array}{rcl}
b_{k,\beta}^{(k,II)}(\varepsilon)&=&\sum_{w\in\Omega_{k}^{(k,II)}(\varepsilon)}
|B_w|^\beta =\sum_{w\in\Omega_{k-1}^{(k-1,I)}(\varepsilon)}
|B_{w\ast((I,II),1,1)}|^\beta\\
&\sim&\zeta_{a_k}^{\beta}\sum_{w\in\Omega_{k-1}^{(k-1,I)}(\varepsilon)}|B_\omega|^\beta
=\zeta_{a_k}^{\beta}\cdot b_{k-1,\beta}^{(k-1,I)}(\varepsilon).
\end{array}
\end{equation}

We  remark that for the three relations above, the constants related
to $``\sim"$  only depend on $V$ and $\beta.$

By iterating \eqref{b-k-1}, \eqref{b-k-3} and \eqref{b-k-2}, we get
\begin{equation}\label{b-k-k-2}
\begin{array}{rcl}
b_{k,\beta}^{(k,I)}(\varepsilon)&\sim&
A_{a_k,\beta}\cdot b_{k-1,\beta}^{(k-1,II)}(\varepsilon) +
A_{a_k,\beta} \cdot b_{k-1,\beta}^{(k-1,III)}(\varepsilon)\\
&\sim& A_{a_k,\beta}\cdot \zeta_{a_{k-1}}^{\beta}\cdot
b_{k-2,\beta}^{(k-2,I)}(\varepsilon)\\
&&+ A_{a_k,\beta}\left(A_{a_{k-1},\beta}\cdot
b_{k-2,\beta}^{(k-2,II)}(\varepsilon) + A_{a_{k-1}-1,\beta} \cdot
b_{k-2,\beta}^{(k-2,III)}(\varepsilon)\right)\\
b_{k,\beta}^{(k,III)}(\varepsilon)&\sim&
A_{a_k,\beta}\cdot b_{k-1,\beta}^{(k-1,II)}(\varepsilon) +
A_{a_k-1,\beta} \cdot b_{k-1,\beta}^{(k-1,III)}(\varepsilon)\\
&\sim& A_{a_k,\beta}\cdot \zeta_{a_{k-1}}^{\beta}\cdot
b_{k-2,\beta}^{(k-2,I)}(\varepsilon)\\
&&+ A_{a_k-1,\beta}\left(A_{a_{k-1},\beta}\cdot
b_{k-2,\beta}^{(k-2,II)}(\varepsilon) + A_{a_{k-1}-1,\beta} \cdot
b_{k-2,\beta}^{(k-2,III)}(\varepsilon)\right)\\
b_{k,\beta}^{(k,II)}(\varepsilon)&\sim&
\zeta_{a_k}^{\beta}\cdot b_{k-1,\beta}^{(k-1,I)}(\varepsilon)\\
&\sim&\zeta_{a_k}^{\beta}A_{a_{k-1},\beta}\left(
b_{k-2,\beta}^{(k-2,II)}(\varepsilon) + b_{k-2,\beta}^{(k-2,III}(\varepsilon)\right).
\end{array}
\end{equation}

We now show that
\begin{equation}\label{b-k-21h}
{b_{k,\beta}^{(k,II)}}(\varepsilon)/{b_{k,\beta}^{(k,I)}(\varepsilon)}
\lesssim{\zeta_{a_k}^{\beta}}/{A_{a_k,\beta}}.
\end{equation}

If $a_{k-1}>1$, by \eqref{b-k-13} we have $b_{k-1,\beta}^{(k-1,I)}(\varepsilon)\sim
b_{k-1,\beta}^{(k-1,III)}(\varepsilon)$. Then by \eqref{b-k-2} we have
$$
\frac{b_{k,\beta}^{(k,II)}(\varepsilon)}{b_{k,\beta}^{(k,I)}(\varepsilon)}\sim
\frac{\zeta_{a_k}^{\beta}\cdot
b_{k-1,\beta}^{(k-1,I)}(\varepsilon)}{b_{k,\beta}^{(k,I)}(\varepsilon)}\lesssim
\frac{\zeta_{a_k}^{\beta}\cdot
b_{k-1,\beta}^{(k-1,III)}(\varepsilon)}{A_{a_k,\beta}\cdot
b_{k-1,\beta}^{(k-1,III)}(\varepsilon)} =\frac{\zeta_{a_k}^{\beta}}{A_{a_k,\beta}}.
$$

If $a_{k-1}=1$, recalling that $\zeta_1=1$ and
$A_{1,\beta}=2^{-\beta}\sim 1$, then by \eqref{b-k-k-2} we get
$$
\begin{cases}
b_{k,\beta}^{(k,I)}(\varepsilon)&
\sim A_{a_k,\beta}\left(b_{k-2,\beta}^{(k-2,I)}(\varepsilon)+
b_{k-2,\beta}^{(k-2,II)}(\varepsilon)\right),\\
b_{k,\beta}^{(k,II)}(\varepsilon)&\sim
\zeta_{a_k}^{\beta}\left(b_{k-2,\beta}^{(k-2,II)}(\varepsilon)+
b_{k-2,\beta}^{(k-2,III)}(\varepsilon)\right).
\end{cases}
$$
By \eqref{b-k-13} we get  $b_{k-2,\beta}^{(k-2,I)}(\varepsilon)
\gtrsim b_{k-2,\beta}^{(k-2,III)}(\varepsilon)$, thus
\eqref{b-k-21h} holds.

We show further that
\begin{equation}\label{b-k-21}
{b_{k,\beta}^{(k,II)}(\varepsilon)}/{b_{k,\beta}^{(k,I)}(\varepsilon)}\sim
{\zeta_{a_k}^{\beta}}/{A_{a_k,\beta}}.
\end{equation}
In fact, by $\zeta_{a_k}\leq 1$, Lemma \ref{A-n-beta} and \eqref{b-k-21h},
we see $b_{k,\beta}^{(k,II)}(\varepsilon)\lesssim b_{k,\beta}^{(k,I)}(\varepsilon)$
for any $k>0$, and also by \eqref{b-k-13}, we have
\begin{eqnarray*}
b_{k,\beta}^{(k,I)}(\varepsilon)&\sim& A_{a_k,\beta}\left(
b_{k-1,\beta}^{(k-1,II)}(\varepsilon) +
b_{k-1,\beta}^{(k-1,III)}(\varepsilon)\right)\\
&\lesssim&A_{a_k,\beta}\left( b_{k-1,\beta}^{(k-1,I)}(\varepsilon) +
 b_{k-1,\beta}^{(k-1,I)}(\varepsilon)\right)\\
&\sim& A_{a_k,\beta} b_{k-1,\beta}^{(k-1,I)}(\varepsilon).
\end{eqnarray*}
Together with \eqref{b-k-2}, we get the other direction of \eqref{b-k-21}.

By \eqref{b-k-13} and \eqref{b-k-21}, we have
\begin{equation}\label{b-k-a1}
b_{k,\beta}(\varepsilon)=b_{k,\beta}^{(k,I)}(\varepsilon)+
b_{k,\beta}^{(k,II)}(\varepsilon)+b_{k,\beta}^{(k,III)}(\varepsilon)
\sim b_{k,\beta}^{(k,I)}(\varepsilon),
\end{equation}
which implies the first and the second formula of \eqref{three},
also the third formula in case of $a_{k}>1$.
If $a_{k}=1$, by \eqref{b-k-3} and \eqref{b-k-21},
$$
\begin{array}{l}
b_{k,\beta}^{(k,III)}(\varepsilon)\sim b_{k-1,\beta}^{(k-1,II)}(\varepsilon)\sim
\frac{\zeta_{a_{k-1}}^{\beta}}{A_{a_{k-1},\beta}}b_{k-1,\beta}^{(k-1,I)}(\varepsilon)\\
b_{k,\beta}(\varepsilon)\sim b_{k,\beta}^{(k,I)}(\varepsilon)\sim
b_{k,\beta}^{(k,II)}(\varepsilon)\sim b_{k-1,\beta}^{(k-1,I)}(\varepsilon)
\end{array}
$$
Thus the third formula of \eqref{three} hold in the case of  $a_k=1$,

Combine \eqref{b-k-1} and \eqref{b-k-a1}.
If $a_{k-1}>1$, \eqref{b-k-13} and \eqref{b-k-21} implies \eqref{b-k-b-k-1}.
If $a_{k-1}=1$, by \eqref{b-k-21}, we see
${b_{k-1,\beta}^{(k-1,II)}}(\varepsilon)/{b_{k-1,\beta}^{(k-1,I)}(\varepsilon)}\sim 1$,
and then we still have \eqref{b-k-b-k-1}.

By Remark \ref{const} and the remark given after the three
relations, all the constants related to $``\sim, \lesssim,\gtrsim"$
only depend on $V$ and $\beta.$
\end{proof}

\begin{prop}
For any  $m\ge k+3$, we have
\begin{equation}\label{cki}
\begin{cases}
\frac{b_{m,\beta}^{(k,I)}(\varepsilon)}{b_{m,\beta}(\varepsilon)}&\sim
\frac{\zeta_{a_{k+1}}^\beta}{A_{a_{k+1},\beta}};\\
\frac{b_{m,\beta}^{(k,II)}(\varepsilon)}{b_{m,\beta}(\varepsilon)}&\sim
\frac{\zeta_{a_k}^\beta}{A_{a_k,\beta}};\\
\frac{b_{m,\beta}^{(k,III)}(\varepsilon)}{b_{m,\beta}(\varepsilon)}&\sim
\begin{cases}
1& a_k>1,a_{k+1}>1;\\
\frac{\zeta_{a_{k+2}}^\beta}{A_{a_{k+2},\beta}}& a_k>1,a_{k+1}=1;\\
\frac{\zeta_{a_{k-1}}^\beta}{A_{a_{k-1},\beta}}& a_k=1,a_{k+1}>1;\\
\frac{\zeta_{a_{k-1}}^\beta}{A_{a_{k-1},\beta}}
\frac{\zeta_{a_{k+2}}^\beta}{A_{a_{k+2},\beta}}&
a_k=1,a_{k+1}=1.
\end{cases}
\end{cases}
\end{equation}

\end{prop}
\begin{proof}
Take any $\sigma_2\in\Omega_{k+1}^{(k,I)}(\varepsilon)$, then $B_{\sigma_2}$ is
a band of type $(k+1,II)$. Take
$\sigma_1,\sigma_3\in\Omega_{k+1}^{(k,II)}(\varepsilon)$ such that
$B_{\sigma_1}$ is a band of type $(k+1,I)$ and $B_{\sigma_3}$ is a
band of type $(k+1,III)$.
For any $p\le m$ and any $T=I,II,III$ define
$$
\Omega_{p,m}^{(T)}(\varepsilon)=\{w_{p}\cdots w_m\in\prod_{j=p}^m \mathscr{E}_{a_j}(\varepsilon)
\text{ admissible }: e_{w_p}=(T,\ast)\},
$$
where $\mathscr{E}_{a_j}(\varepsilon)$ is defined in \eqref{trunc}.

Define
$$\begin{array}{rcl}
c_{k+1}^{(I)}&=&\Sigma_{\tau\in\Omega_{k+2,m}^{(I)}(\varepsilon)}
\frac{|B_{\sigma_1*\tau}|^\beta}{|B_{\sigma_1}|^\beta}\\
c_{k+1}^{(II)}&=&\Sigma_{\tau\in\Omega_{k+2,m}^{(II)}(\varepsilon)}
\frac{|B_{\sigma_2*\tau}|^\beta}{|B_{\sigma_2}|^\beta}\\
c_{k+1}^{(III)}&=&\Sigma_{\tau\in\Omega_{k+2,m}^{(III)}(\varepsilon)}
\frac{|B_{\sigma_3*\tau}|^\beta}{|B_{\sigma_3}|^\beta}.
\end{array}$$
We can also define $c_{k+2}^{(I)}$,  $c_{k+2}^{(II)}$,
$c_{k+2}^{(III)}$ in an analogous way.
Analogous to the arguments of \eqref{b-k-1},\eqref{b-k-3} and \eqref{b-k-2}, we have
\begin{eqnarray*}
c_{k+1}^{(I)}&\sim& \zeta_{a_{k+2}}^\beta c_{k+2}^{(II)}\\
c_{k+1}^{(II)}&\sim&
A_{a_{k+2}+1,\beta}\ c_{k+2}^{(I)}+A_{a_{k+2},\beta}\ c_{k+2}^{(III)}\\
c_{k+1}^{(III)}&\sim&
A_{a_{k+2},\beta}\ c_{k+2}^{(I)}+A_{a_{k+2}-1,\beta}\ c_{k+2}^{(III)}.
\end{eqnarray*}
And consequently
$$
\frac{c_{k+1}^{(I)}}{c_{k+1}^{(II)}}\sim\frac{ \zeta_{a_{k+2}}^\beta
}{A_{a_{k+2},\beta}};\ \ \frac{c_{k+1}^{(III)}}{c_{k+1}^{(II)}}\sim
\begin{cases}
1&a_{k+2}>1\\
\frac{ \zeta_{a_{k+3}}^\beta }{A_{a_{k+3},\beta}}& a_{k+2}=1
\end{cases}
$$

Write
$$
\rho:=(e_{12},1,1),\ \  \theta_j:=(e_{21},a_{k+1}+1,j)\ \ \text{ and }\ \
\phi_j:=(e_{23},a_{k+1},j).
$$ By
Theorem \ref{bco} and Corollary \ref{C-n} we have
\begin{eqnarray*}
b_{m,\beta}^{(k,I)}(\varepsilon)&=&\sum_{w\in\Omega_{m}^{(k,I)}(\varepsilon)}|B_{w}|^\beta\\
&=&\sum_{\sigma\in\Omega_{k}^{(k,I)}(\varepsilon)}
|B_{\sigma}|^\beta\frac{|B_{\sigma*\rho}|^\beta}{|B_{\sigma}|^\beta}
\sum_{\tau\in\Omega_{k+2,m}^{(k+1,II)}(\varepsilon)}\frac{|B_{\sigma*\rho*\tau}
|^\beta}{|B_{\sigma*\rho}|^\beta}\\
&\sim&\sum_{\sigma\in\Omega_{0,k}^{(k,I)}(\varepsilon)}|B_{\sigma}|^\beta
\frac{|B_{\sigma_2}|^\beta}{|B_{\sigma_2^-}|^\beta}
\sum_{\tau\in\Omega_{k+2,m}^{(k+1,II)}(\varepsilon)}
\frac{|B_{\sigma_2*\tau}|^\beta}{|B_{\sigma_2}|^\beta}\\
&\sim&b_{k,\beta}^{(k,I)}(\varepsilon) \zeta_{a_{k+1}}^\beta c_{k+1}^{(II)},\\
b_{m,\beta}^{(k,II)}(\varepsilon)&=&
\Sigma_{\sigma\in\Omega_{0,m}^{(k,II)}(\varepsilon)}|B_{\sigma}|^\beta\\
&=&\Sigma_{\sigma\in\Omega_{0,k}^{(k,II)}(\varepsilon)}|B_{\sigma}|^\beta
\Big(\sum_{j=\lceil \varepsilon(a_{k+1}+2)\rceil}^{\lfloor(1-\varepsilon)(a_{k+1}+2)\rfloor}
\Sigma_{\tau\in\Omega_{k+2,m}^{(k+1,I)}(\varepsilon)}
\frac{|B_{\sigma*\theta_j*\tau}|^\beta}{|B_{\sigma}|^\beta}\\
&&\ \ \ \ \ \ \ \ \ \ \ \ \ \ \ \ \ \ \
+\sum_{j=\lceil \varepsilon (a_{k+1}+1)\rceil}^{\lfloor(1-\varepsilon)(a_{k+1}+1)\rfloor}
\Sigma_{\tau\in\Omega_{k+2,m}^{(k+1,III)}(\varepsilon)}
\frac{|B_{\sigma*\phi_j*\tau}|^\beta}{|B_{\sigma}|^\beta}\Big)\\
&\sim&\Sigma_{\sigma\in\Omega_{0,k}^{(k,II)}(\varepsilon)}|B_{\sigma}|^\beta\left(
A_{a_{k+1}+1,\beta}\ c_{k+1}^{(I)}+
A_{a_{k+1},\beta}\ c_{k+1}^{(III)}\right)\\
&\sim&b_{k,\beta}^{(k,II)}(\varepsilon)
A_{a_{k+1},\beta}(c_{k+1}^{(I)}+c_{k+1}^{(III)}),\\
b_{m,\beta}^{(k,III)}(\varepsilon)&\sim&b_{k,\beta}^{(k,III)}(\varepsilon)\left(
A_{a_{k+1},\beta}\ c_{k+1}^{(I)}+ A_{a_{k+1}-1,\beta}\ c_{k+1}^{(III)}\right),
\end{eqnarray*}
where $\sigma_2^-$ is the word obtained by deleting the last letter of $\sigma_2.$

We claim that $c_{k+1}^{(I)}+c_{k+1}^{(III)}\sim c_{k+1}^{(II)}.$ In
fact at first we note that $\zeta_n^\beta/A_{n,\beta}\lesssim 1$ for
any $n\geq 1.$ Thus
\begin{eqnarray*}
c_{k+1}^{(I)}+c_{k+1}^{(III)}&\sim&\left(\frac{
\zeta_{a_{k+2}}^\beta }{A_{a_{k+2},\beta}}+
\begin{cases}1,&\text{ if }a_{k+2}>1\\
\zeta_{a_{k+3}}^\beta/A_{a_{k+3},\beta},& \text{ if }a_{k+2}=1
\end{cases}
\right) c_{k+1}^{(II)}\\
&\sim&
\begin{cases}
\left(\frac{ \zeta_{a_{k+2}}^\beta }{A_{a_{k+2},\beta}}+1\right)
c_{k+1}^{(II)},&\text{ if }a_{k+2}>1\\
\left(\frac{ \zeta_{a_{k+3}}^\beta
}{A_{a_{k+3},\beta}}+2^\beta\right)c_{k+1}^{(II)},&\text{ if } a_{k+2}=1
\end{cases}\\
&\sim&c_{k+1}^{(II)}.
\end{eqnarray*}

Write $\Theta_k:=b_{k,\beta}^{(k,I)}c_{k+1}^{(II)}$. As a result we
get
\begin{eqnarray*}
b_{m,\beta}^{(k,I)}(\varepsilon)&\sim& \zeta_{a_{k+1}}^\beta \Theta_k;\\
b_{m,\beta}^{(k,II)}(\varepsilon)&\sim&\frac{ \zeta_{a_{k}}^\beta
}{A_{a_k,\beta}}  A_{a_{k+1},\beta}\Theta_k;\\
b_{m,\beta}^{(k,III)}(\varepsilon)&\sim&
\begin{cases}
 A_{a_{k+1},\beta}\Theta_k & a_k>1;a_{k+1}>1;\\
\frac{ \zeta_{a_{k+2}}^\beta }{A_{a_{k+2},\beta}}\Theta_k & a_k>1;a_{k+1}=1;\\
\frac{ \zeta_{a_{k-1}}^\beta }{A_{a_{k-1},\beta}} A_{a_{k+1},\beta}\Theta_k & a_k=1;a_{k+1}>1;\\
\frac{ \zeta_{a_{k-1}}^\beta }{A_{a_{k-1},\beta}}
\frac{ \zeta_{a_{k+2}}^\beta }{A_{a_{k+2},\beta}}\Theta_k & a_k=1;a_{k+1}=1.\\
\end{cases}
\end{eqnarray*}
By a simple computation the result follows.
\end{proof}

We will prove the following theorem,
which is Theorem \ref{gibbs0} when $\varepsilon=0$.

\begin{theo}[Existence of Gibbs like measures]\label{gibbs}
For any $0<\beta<1$, $0\le \varepsilon<1/12$, there exists a probability measure $\mu_{\beta,\varepsilon}$
supported on $E_\varepsilon$ such that

 if $w\in\Omega_{k}^{(k,I)}(\varepsilon)$, let $u=(e_{12},1,1)$,  then
\begin{equation}\label{gibbs1}
\mu_{\beta,\varepsilon}(B_{w})\sim\frac{\zeta_{a_{k+1}}^\beta}{ a_{k+1}^{1-\beta}}
\frac{|B_{w}|^\beta}{b_{k,\beta}(\varepsilon)}\sim
\frac{|B_{wu}|^\beta}{b_{k+1,\beta}(\varepsilon)}.
\end{equation}

If $w\in\Omega_{k}^{(k,II)}(\varepsilon)$, then
\begin{equation}\label{gibbs2}
\mu_{\beta,\varepsilon}(B_{w})\sim
\frac{|B_{w}|^\beta}{b_{k,\beta}(\varepsilon)}.
\end{equation}

If $w\in\Omega_{k}^{(k,III)}(\varepsilon)$, then
\begin{equation}\label{gibbs3}
\mu_{\beta,\varepsilon}(B_{w})\sim
\begin{cases}
\frac{|B_{w}|^\beta}{b_{k,\beta}(\varepsilon)}& a_{k+1}>1;\\
\frac{\zeta_{a_{k+2}}^\beta}{a_{k+2}^{1-\beta}}
\frac{|B_{w}|^\beta}{b_{k,\beta}(\varepsilon)}& a_{k+1}=1.
\end{cases}
\end{equation}
\end{theo}

\begin{proof}
For any $0<\beta<1$ and $m>0$, we define a probability
$\mu_{\beta,\varepsilon,m}$ on $\mathbb R$  such that for any
$w\in\Omega_{m}(\varepsilon)$,
$$\mu_{\beta,\varepsilon,m}(B_{w})=
\frac{|B_{w}|^\beta}{b_{m,\beta}(\varepsilon)},$$
where $\mu_{\beta,\varepsilon,m}$ is uniformly distributed on each band
$B_w$ for any $w\in\Omega_{m}(\varepsilon)$.

Fix any $k\ge 1.$ For  $T\in\{I,II,III\}$ and $w \in \Omega_{k}^{(k,T)}(\varepsilon)$,
We will prove that
$\mu_{\beta,\varepsilon,m}(B_{w})$ satisfy \eqref{gibbs1},\eqref{gibbs2}
or \eqref{gibbs3} respectively for $m\ge k+3$. Then by taking any weak limit of
$\{\mu_{\beta,\varepsilon,m}\}_{m>0}$, we prove the theorem.

For any $\sigma\in\Omega_{k}^{(k,T)}(\varepsilon)$, by bounded covariation  we have
$$\begin{array}{rcl}
\mu_{\beta,\varepsilon,m}(B_{w})&=&
\frac{1}{b_{m,\beta}(\varepsilon)}
\sum\limits_{\tau\in\Omega_{k+1,m}^{(k,T)}(\varepsilon)}|B_{w*\tau}|^\beta\\
&=&\frac{|B_{w}|^\beta}
{b_{m,\beta}(\varepsilon)}\sum\limits_{\tau\in\Omega_{k+1,m}^{(k,T)}(\varepsilon)}
\frac{|B_{w*\tau}|^\beta}{|B_{w}|^\beta}\\
&\sim&\frac{|B_{w}|^\beta}
{b_{m,\beta}(\varepsilon)}\sum\limits_{\tau\in\Omega_{k+1,m}^{(k,T)}(\varepsilon)}
\frac{|B_{\sigma*\tau}|^\beta}{|B_{\sigma}|^\beta}.
\end{array}
$$
Hence
$$|B_{\sigma}|^\beta \mu_{\beta,\varepsilon,m}(B_{w})\sim\frac{|B_{w}|^\beta}
{b_{m,\beta}(\varepsilon)}
\sum\limits_{\tau\in\Omega_{k+1,m}^{(k,T)}(\varepsilon)}|B_{\sigma*\tau}|^\beta.$$
Take sum on both sides for all $\sigma\in\Omega_{k}^{(k,T)}(\varepsilon)$, we get
$$b_{k,\beta}^{(k,T)}(\varepsilon)\mu_{\beta,\varepsilon,m}(B_{w})
\sim|B_{w}|^\beta \frac{b_{m,\beta}^{(k,T)}(\varepsilon)}{b_{m,\beta}(\varepsilon)},$$
which implies that
$$\mu_{\beta,\varepsilon,m}(B_{w})\sim\frac{|B_{w}|^\beta}{b_{k,\beta}(\varepsilon)}
\frac{b_{k,\beta}(\varepsilon)}{b_{k,\beta}^{(k,T)}(\varepsilon)}
\frac{b_{m,\beta}^{(k,T)}(\varepsilon)}{b_{m,\beta}(\varepsilon)}.$$ Combining with
\eqref{three} and \eqref{cki},

 if $w$ has type $I$ then
$$
\mu_{\beta,\varepsilon,m}(B_{w})\sim\frac{\zeta_{a_{k+1}}^\beta}{A_{a_{k+1},\beta}}
\frac{|B_{w}|^\beta}{b_{k,\beta}(\varepsilon)}.
$$

If $w$ has type $II$, then
$$
\mu_{\beta,\varepsilon,m}(B_{w})\sim
\frac{|B_{w}|^\beta}{b_{k,\beta}(\varepsilon)}.
$$

If $w$ has type $III$, then
$$
\mu_{\beta,\varepsilon,m}(B_{w})\sim
\begin{cases}
\frac{|B_{w}|^\beta}{b_{k,\beta}(\varepsilon)}& a_{k+1}>1;\\
\frac{\zeta_{a_{k+2}}^\beta}{A_{a_{k+2},\beta}}
\frac{|B_{w}|^\beta}{b_{k,\beta}(\varepsilon)}& a_{k+1}=1.
\end{cases}
$$
Thus we get \eqref{gibbs2}, \eqref{gibbs3} and the first relation of \eqref{gibbs1}.

 To get the second relation of \eqref{gibbs1}, we proceed as follows.
Write $u=(e_{12},1,1).$ If $w\in \Omega_{k}^{(k,I)}(\varepsilon)$, then
$w*u\in \Omega_{k+1}^{(k+1,II)}(\epsilon)$ and
$\mu(B_{w*u})=\mu(B_w).$ Now the result follows   by  applying
\eqref{gibbs2}.
\end{proof}

\noindent
{\bf Acknowledgements}. Liu and Qu are supported by the National Natural Science Foundation of China, No. 11371055.  Qu is supported by the National Natural Science Foundation of China, No. 11201256.   Wen is supported by the National Natural Science Foundation of China, No.11271223. The authors thank Morningside Center of Mathematics for partial support.


\begin{thebibliography}{30}

\bibitem{BIST} J. Bellissard, B. Iochum, E. Scoppola and D. Testart,
Spectral properties of one dimensional quasi-crystals, {\em
Commun. Math. Phys.} {\bf 125}(1989), 527-543.



 \bibitem{C}S. Cantat, Bers and H\'enon, Painlev\'e and Schro\"dinger, Duke Math. J. 149 (2009), 411-460.



\bibitem{DEGT} D.~Damanik, M.~Embree, A.~Gorodetski, and
S.~Tcheremchantsev, the fractal dimension of the spectrum of the
Fibonacci Hamiltonian, {\em Commun. Math. Phys.}{\bf 280:2}(2008), 499-516.

\bibitem{DG}D. Damanik, A. Gorodetski, Hyperbolicity of the trace map for the weakly coupled Fibonacci
Hamiltonian, Nonlinearity 22 (2009), 123-143.

\bibitem{DG2} D. Damanik, A. Gorodetski, Spectral and quantum dynamical properties of the weakly coupled Fibonacci Hamiltonian, Commun. Math. Phys. 305 (2011), 221-277.

\bibitem{DKL} D. Damanik, R. Killip, D. Lenz,
{Uniform spectral properties of one-dimensional quasicrystals,
III. $\alpha$-continuity}, {\em Commun. Math. Phys.} {\bf 212},
(2000),191-204.

\bibitem{Fa} K. Falconer, techniques in fractal geometry,
John Wiley\& Sons, 1997.

\bibitem{FLW} S. Fan, Q.H. Liu, Z.Y. Wen,
Gibbs like measure for spectrum of a class of quasi-crystals,
{\em Ergodic Theory Dynam. Systems}, {\bf 31}(2011), 1669-1695.

\bibitem{FWW} D.-J. Feng, Z.-Y. Wen and J. Wu,
{\it Some dimensional results for homogeneous Moran sets},
{\em
Science of China}(Series A), {\bf 40:5}(1997), 475-482.



\bibitem{JL} S. Jitomirskaya and Y. Last, Power-law subordinacy and singular spectra. II. Line operators,
Commun. Math. Phys. 211 (2000), 643-658.

\bibitem{KKT} M. Kohmoto, L. P. Kadanoff, C. Tang, {\it Localization problem in one dimension: mapping
and escape}, {\em Phys. Rev. Lett.} {\bf 50} (1983), 1870-1872.



\bibitem{LPW07}Q.H. Liu, J. Peyri\`ere and Z.Y. Wen,
Dimension of the spectrum of one-dimensional discrete
Schrodinger operators with Sturmian potentials,
{\em Comptes Randus Mathematique} {\bf 345:12}(2007), 667--672.


\bibitem{LW} Q.H. Liu, Z.Y. Wen,
Hausdorff dimension of spectrum of one-dimensional
Schr\"odinger operator with Sturmian potentials,
{\em Potential Analysis} {\bf 20:1}(2004), 33--59.

\bibitem{LW05} Q.H. Liu, Z.Y. Wen,
On dimensions of multitype Moran sets,
{\em Math. Proc. Camb. Phyl. Soc.} {\bf 139:3}(2005), 541--553.

\bibitem{MRW} J.H. Ma, H. Rao, Z.Y. Wen, \textit{Dimensions of Cookie-cutter-like sets},
{\em Science of China}(Series A),
{\bf 44:11}(2001), 1400-1412.

\bibitem{OPRSS} S. Ostlund, R. Pandit, D. Rand, H. Schellnhuber, E. Siggia, One-dimensional Schr\"odinger equation with an almost periodic potential, { \em Phys. Rev. Lett.} 50 (1983), 1873-1877.

\bibitem{R} L. Raymond,
{\em A constructive gap labelling for the discrete schr\"odinger
operater on a quasiperiodic chain}.(Preprint,1997)



\bibitem{Su} A. S\"ut\"o,
, Singular continuous spectrum on a Cantor set of zero Lebesgue measure for the Fibonacci Hamiltonian, J. Stat. Phys. 56 (1989), 525-531.

\bibitem{T} M. Toda,
{\em Theory of Nonlinear Lattices},
Number 20 in Solid-State Sciences,
Springer-Verlag, second enlarged edition, 1989. Chap. 4.

\bibitem{Tr} C. Tricot, Douze definitions de la densite logarithmique,
C. R. Acad. Sc. Paris, 1981, 293: 549.

\end{thebibliography}
\end{document}